\documentclass[runningheads]{llncs}

\usepackage{amssymb, amsmath, hyperref}
\usepackage{synttree, bussproofs,scrextend, stmaryrd, rotating, xcolor}
\usepackage{esvect, comment}
\usepackage{pgf, tikz, color}
\usetikzlibrary{arrows, automata}
\usepackage[all]{xy}
\usepackage{changepage}
\usepackage{enumerate}
\usepackage{proof,upgreek,hyperref}
\usepackage{marvosym,amssymb}

\def\imp{\supset}

\def\far{\rightarrow}

\def\wk{(wk)}

\def\cut{(cut)}
\def\id{(id)}
\def\idfo{(id_{q})}

\def\botr{(\bot_{l})}

\newcommand{\disrr}{(\vee_{r})}
\newcommand{\conrr}{(\wedge_{r})}
\newcommand{\disrl}{(\vee_{l})}
\newcommand{\conrl}{(\wedge_{l})}
\newcommand{\imprr}{(\supset_{r})}
\newcommand{\imprl}{(\supset_{l})}
\newcommand{\trans}{(tra)}
\newcommand{\refl}{(ref)}

\def\lift{(lift)}

\def\allr{(\forall_{r})}
\def\alll{(\forall_{l})}
\def\existsl{(\exists_{l})}
\def\existsr{(\exists_{r})}
\def\nd{(nd)}
\def\cd{(cd)}

\def\negl{(\neg_{l})}
\def\negr{(\neg_{r})}

\def\idnew{(id^{*})}
\def\idfonew{(id^{*}_{q})}
\def\allnewl{(\forall_{l}^{*})}

\def\intfocd{\mathsf{IntQC}}

\def\int{\mathsf{Int}}

\def\lint{\mathsf{G3Int}}

\def\intfocdl{\mathsf{IntQCL}}
\def\nint{\mathsf{NInt}}

\def\lintfocd{\mathsf{G3IntQC}}

\def\nintfocd{\mathsf{NIntQC}}

\def\lcut{\{}
\def\rcut{\}}

\def\switch{\mathfrak{N}}

\providecommand{\acknowledgments}[1]{\textbf{Acknowledgments. } #1}

\newcommand{\R}{\mathcal{R}}

\newtheorem{innercustomlem}{Lemma}
\newenvironment{customlem}[1]
  {\renewcommand\theinnercustomlem{#1.}\innercustomlem}
  {\endinnercustomlem}
  
\newtheorem{innercustomthm}{Theorem}
\newenvironment{customthm}[1]
  {\renewcommand\theinnercustomthm{#1.}\innercustomthm}
  {\endinnercustomthm}

\begin{document}

\title{On Deriving Nested Calculi for Intuitionistic Logics from Semantic Systems\thanks{Work funded by FWF projects I2982, Y544-N2, and W1255-N23.}}


\author{Tim Lyon} 

\institute{Institut f\"ur Logic and Computation, Technische Universit\"at Wien, 1040 Wien, Austria  \\ \email{lyon@logic.at}}

\authorrunning{Tim Lyon}

\maketitle

\begin{abstract}
This paper shows how to derive nested calculi from labelled calculi for propositional intuitionistic logic and first-order intuitionistic logic with constant domains, thus connecting the general results for labelled calculi with the more refined formalism of nested sequents. The extraction of nested calculi from labelled calculi obtains via considerations pertaining to the elimination of structural rules in labelled derivations. Each aspect of the extraction process is motivated and detailed, showing that each nested calculus inherits favorable proof-theoretic properties from its associated labelled calculus.  

\end{abstract}

\section{Introduction}

Numerous fruitful consequences and applications obtain through the supplementation of a logic with an \emph{analytic} calculus. Such calculi are characterized on the basis of their inference rules, which stepwise (de)compose the formula to be proven. One of the most prominent realizations of this idea dates back to Gentzen~\cite{Gen35}, who proposed the \emph{sequent calculus} framework for classical and intuitionistic logic. Since then, countless extensions and reformulations of Gentzen's framework have been supplied for many logics of interest. Examples of extensions include \emph{display calculi}~\cite{Bel82,TiuIanGor12}, \emph{hypersequent calculi}~\cite{Pog08}, \emph{labelled calculi}~\cite{DycNeg12,Vig00}, and \emph{nested calculi}~\cite{Fit14,TiuIanGor12}. Such calculi have been exploited to prove meaningful results; e.g. decidability~\cite{Pog08,TiuIanGor12}, interpolation~\cite{LyoTiuGorClo20}, and automated counter-model extraction~\cite{LyoBer19,TiuIanGor12}. We focus on the labelled and nested formalisms in this paper.

The labelled approach of constructing calculi may be qualified as \emph{semantic} due to the fact that calculi are obtained through the transformation of semantic clauses and Kripke-frame properties into inference rules for a logic~\cite{DycNeg12,Vig00}. Although the approach has been criticized by some~\cite{Avr96}, it has also proven to be quite successful relative to certain criteria. For example, the labelled formalism is surprisingly modular and allows for the automated construction of analytic calculi for many intermediate and modal logics~\cite{CiaMafSpe13,DycNeg12,Neg16}. Furthermore, calculi constructed in the labelled paradigm often possess fruitful properties (e.g. contraction-admissibility, invertibility of rules, cut-elimination, etc.) that follow from general results~\cite{DycNeg12,Neg16}.

In 2009, Br\"unnler introduced \emph{nested sequent calculi}~\cite{Bru09} and Poggiolesi introduced \emph{tree hypersequent calculi}~\cite{Pog09Trends} for a set of modal logics. Both formalisms are essentially notational variants of one another and make use of an idea due to Bull~\cite{Bul92} and Kashima~\cite{Kas94} to organize sequents into treelike structures called \emph{nested sequents}. Although nested sequents can be seen as a distinct proof-theoretic formalism, it was shown in 2012~\cite{Fit12} that nested sequent calculi can be viewed as `upside-down' versions of prefixed tableaux, introduced much earlier in 1972~\cite{Fit72}.
In contrast to labelled sequents, nested sequents 
are often given in a language as expressive as the language of the logic; thus, nested calculi have the advantage that they minimize the bureaucracy sufficient to prove theorems. The nested formalism continues to receive much attention, proving itself suitable for constructing analytic calculi~\cite{Bru09}, developing automated reasoning algorithms~\cite{GorPosTiu11}, and verifying interpolation~\cite{LyoTiuGorClo20}, among other applications.

Despite the many advantages of nested calculi, constructing such calculi for logics as well as proving that they possess favorable proof-theoretic properties (admissibility of structural rules, cut-elimination, etc.) is often done on a case by case basis; i.e. the nested formalism does not---to date---offer the same generality of results that hold in the labelled paradigm (cf.~\cite{DycNeg12,Neg16}). 
Therefore, a significant advantage of the labelled paradigm over the nested paradigm is that labelled calculi are easily constructed on the basis of a logic's semantics and one often obtains highly favorable proof-theoretic properties of the calculi (essentially) for free via general theorems~\cite{DycNeg12,Neg16}. Nevertheless, the labelled formalism has its drawbacks: the calculi involve a complicated syntax, and labelled structural rules typically delete vital formulae from premise to conclusion, which can cause associated proof-search algorithms to be less efficient or rely on backtracking.

Since the labelled formalism is well-suited for constructing calculi and confirming properties, and the nested formalism is well-suited for applications, a general method of extracting nested calculi from labelled calculi (with properties preserved) is highly desirable. One could generate labelled calculi for a class of logics and confirm favorable proof-theoretic properties via existing general results; if such properties were preserved during an extraction procedure, then the ensuing nested calculi would possess the properties as well, yielding 
practical, cut-free nested calculi. 
Similar ideas and results have been discussed in the literature~\cite{CiaLyoRam18,GorRam12AIML,LyoBer19,Pim18}, where refined calculi (which can be viewed as nested calculi) were extracted from labelled calculi for various logics. (NB. The author has recently been made aware of~\cite{Pim18}, which mentions results strongly related to Sect.~\ref{section-3}. Although the results presented here were discovered independently, the work of Sect.~\ref{section-3} can be seen as a detailed explication and expansion of the work presented in~\cite{Pim18}.). In this paper, we advance our understanding of this method, and show how to derive Fitting's nested calculi (see~\cite{Fit14}) from labelled calculi for intuitionistic logics. The results of this paper are also worthwhile in that they clarify the connection between the intuitionistic labelled calculi and nested calculi considered, thus shedding light on the semantic roles played by certain inference rules and syntactic structures in Fitting's formalism.

This paper is organized as follows: Sect.~\ref{section-2} introduces the labelled and nested calculi for the intuitionistic logics considered, and Sect.~\ref{section-notation} shows how to translate labelled sequents into nested sequents. Sect.~\ref{section-3} and~\ref{section-4} show how to extract the nested calculi from the labelled calculi for propositional and first-order intuitionistic logic with constant domains, respectively. Last, Sect.~\ref{conclusion} concludes and discusses future work.



\section{Proof Calculi for Intuitionistic Logics}\label{section-2}

The language $\mathcal{L}$ for propositional intuitionistic logic ($\int$) is defined via the BNF grammar shown below top, and the language $\mathcal{L_{Q}}$ for constant domain first-order intuitionistic logic ($\intfocd$) is defined via the BNF grammar shown below bottom:
$$A ::= p \ | \ \bot \ | \ (A \lor A) \ | \ (A \land A) \ | \ (A \imp A)$$
$$A ::= p(x_{1}, \ldots, x_{n}) \ | \ \bot \ | \ (A \vee A) \ | \ (A \wedge A) \ | \ (A \imp A) \ | \ (\forall x) A \ | \ (\exists x) A $$
In the language $\mathcal{L}$, $p$ is among a denumerable set of \emph{propositional variables} $\{p, q, r, \ldots\}$. In the language $\mathcal{L_{Q}}$, $p$ is an $n$-ary \emph{predicate symbol} with $x_{1}, \ldots, x_{n}, x$ \emph{variables} ($n \in \mathbb{N}$), and when $n = 0$, $p$ is assumed to be a propositional variable. 
As usual, we define $\neg A := A \imp \bot$.

We assume the reader is familiar with intuitionistic logics; for a comprehensive overview, see~\cite{GabSheSkv09}.

\subsection{The Labelled Calculi $\lint$ and $\lintfocd$}

We define \emph{propositional $($first-order$)$ labelled sequents} to be syntactic objects of the form $L_{1} \Rightarrow L_{2}$ ($L_{1}' \Rightarrow L_{2}'$, resp.), where $L_{1}$ and $L_{2}$ ($L_{1}'$ and $L_{2}'$, resp.) are formulae defined via the BNF grammar below top (below bottom, resp.). 
\begin{center}
\begin{tabular}{c @{\hskip 2em} c}
$L_{1} ::= w : A \ | \ w \leq v \ | \ L_{1},L_{1}$

&

$L_{2} ::= w : A \ | \ L_{2},L_{2}$
\end{tabular}
\end{center}
\begin{center}
\begin{tabular}{c @{\hskip 2em} c}
$L_{1}' ::= w : A \ | \ a \in D_{w} \ | \ w \leq v \ | \ L_{1}',L_{1}'$

&

$L_{2}' ::= w : A \ | \ L_{2}',L_{2}'$
\end{tabular}
\end{center}
In the propositional case, $A$ is in the language $\mathcal{L}$ and $w$ is among a denumerable set of labels $\{w, v, u, \ldots \}$. In the first-order case, $A$ is in the language $\mathcal{L_{Q}}$, $a$ is among a denumerable set of \emph{parameters} $\{a, b, c, \ldots\}$, and $w$ is among a denumerable set of labels $\{w, v, u, \ldots \}$. We refer to formulae of the forms $w \leq u$ and $a \in D_{w}$ as \emph{relational atoms} (with formulae of the form $a \in D_{w}$ sometimes referred to as \emph{domain atoms}, more specifically) and refer to formulae of the form $w : A$ as \emph{labelled formulae}. Due to the two types of formulae occurring in a labelled sequent, we often use $\R$ to denote relational atoms, and $\Gamma$ and $\Delta$ to denote labelled formulae, thus distinguishing between the two. Labelled sequents are therefore written in a general form as $\R, \Gamma \Rightarrow \Delta$. Moreover, we take the comma operator to be commutative and associative; for example, we identify the formula $w : A, w \leq u, u:B$ with $w \leq u, u :B, w : A$. This interpretation of comma lets us view $\R, \Gamma$ and $\Delta$ as multisets. Also, we allow for empty antecedents and succedents in both our labelled and nested sequents.

In the first-order setting, we syntactically distinguish between \emph{bound variables} $\{x, y, z, \ldots\}$ and \emph{free variables}, which are replaced with \emph{parameters} $\{a, b, c, \ldots\}$, to avoid clashes between the two categories (cf.~\cite[Sect.~8]{Fit14}). Therefore, instead of using formulae directly from the first-order language, we use formulae from the first-order language where each freely occurring variable $x$ has been replaced by a distinct parameter $a$. For example, we would make use of the labelled formula $w : (\forall x) p(a,x) \lor q(a,b)$ instead of $w : (\forall x) p(y,x) \lor q(y,z)$ in a first-order sequent of $\lintfocd$. For a formula $A \in \mathcal{L_{Q}}$, we write $A[a/x]$ to mean the formula that results from substituting the parameter $a$ for all occurrences of the free variable $x$ in $A$. Last, we use the notation $A(a_{0}, \ldots, a_{n})$, with $n \in \mathbb{N}$, to denote that the parameters $a_{0}, \ldots, a_{n}$ are all parameters occurring in the formula $A$. We write $A(\vv{a})$ as shorthand for $A(a_{0}, \ldots, a_{n})$ and $\vv{a} \in D_{w}$ as shorthand for $a_{0} \in D_{w}, \ldots, a_{n} \in D_{w}$. The labelled calculi are given in Fig.~\ref{fig:propositional-calculus}.

\begin{figure}
\noindent\hrule

\begin{center}
\begin{tabular}{c @{\hskip 1em} c} 

\AxiomC{}
\RightLabel{$\id$}
\UnaryInfC{$\R,w \leq v,w :p,\Gamma \Rightarrow \Delta, v :p$}
\DisplayProof

&

\AxiomC{$\R, w \leq v, v :A, \Gamma \Rightarrow \Delta, v :B$}
\RightLabel{$\imprr^{\dag_{1}}$}
\UnaryInfC{$\R, \Gamma \Rightarrow \Delta, w :A \imp B$}
\DisplayProof

\end{tabular}
\end{center}

\begin{center}
\begin{tabular}{c @{\hskip 1em} c @{\hskip 1em} c}

\AxiomC{$\R,w :A, w :B, \Gamma \Rightarrow \Delta$}
\RightLabel{$\conrl$}
\UnaryInfC{$\R, w :A \wedge B, \Gamma \Rightarrow \Delta$}
\DisplayProof

&

\AxiomC{$\R, \Gamma \Rightarrow \Delta, w :A$}
\AxiomC{$\R, \Gamma \Rightarrow \Delta, w :B$}
\RightLabel{$\conrr$}
\BinaryInfC{$\R, \Gamma \Rightarrow \Delta, w :A \wedge B$}
\DisplayProof

\end{tabular}
\end{center}

\begin{center}
\begin{tabular}{c @{\hskip 1em} c @{\hskip 1em} c}

\AxiomC{$\R, w :A, \Gamma \Rightarrow \Delta$}
\AxiomC{$\R, w :B, \Gamma \Rightarrow \Delta$}
\RightLabel{$\disrl$}
\BinaryInfC{$\R, w :A \vee B, \Gamma \Rightarrow \Delta$}
\DisplayProof

&

\AxiomC{$\R, \Gamma \Rightarrow \Delta, w :A, w :B$}
\RightLabel{$\disrr$}
\UnaryInfC{$\R, \Gamma \Rightarrow \Delta, w :A \vee B$}
\DisplayProof

\end{tabular}
\end{center}




\begin{center}
\begin{tabular}{c}
\AxiomC{$\R,w \leq v, w :A \imp B, \Gamma \Rightarrow \Delta, v :A$}
\AxiomC{$\R,w \leq v, w :A \imp B, v :B, \Gamma \Rightarrow \Delta$}
\RightLabel{$\imprl$}
\BinaryInfC{$\R,w \leq v, w :A \imp B, \Gamma \Rightarrow \Delta$}
\DisplayProof 
\end{tabular}
\end{center}

\begin{center}
\begin{tabular}{c @{\hskip 1em} c}

\AxiomC{$\R,w \leq w, \Gamma \Rightarrow \Delta$}
\RightLabel{$\refl$}
\UnaryInfC{$\R,\Gamma \Rightarrow \Delta$}
\DisplayProof

&

\AxiomC{$\R,w \leq v, v \leq u, w \leq u, \Gamma \Rightarrow \Delta$}
\RightLabel{$\trans$}
\UnaryInfC{$\R,w \leq v, v \leq u, \Gamma \Rightarrow \Delta$}
\DisplayProof

\end{tabular}
\end{center}

\begin{center}
\begin{tabular}{c @{\hskip 1em} c}
\AxiomC{}
\RightLabel{$\botr$}
\UnaryInfC{$\R,w :\bot, \Gamma \Rightarrow \Delta$}
\DisplayProof

&

\AxiomC{}
\RightLabel{$\idfo$}
\UnaryInfC{$\R,w \leq v, \vv{a} \in D_{w}, w :p(\vv{a}),\Gamma \Rightarrow \Delta, v :p(\vv{a})$}
\DisplayProof
\end{tabular}
\end{center}

\begin{center}
\scalebox{.95}{
\begin{tabular}{c c} 

\AxiomC{$\R, w \leq v, a \in D_{v}, \Gamma \Rightarrow \Delta, v : A[a/x]$}
\RightLabel{$\allr^{\dag_{2}}$}
\UnaryInfC{$\R, \Gamma \Rightarrow \Delta, w : \forall x A$}
\DisplayProof

&

\AxiomC{$\R, a \in D_{w}, \Gamma \Rightarrow \Delta, w: A[a/x], w: \exists x A$}
\RightLabel{$\existsr$}
\UnaryInfC{$\R, a \in D_{w}, \Gamma \Rightarrow \Delta, w: \exists x A$}
\DisplayProof

\end{tabular}
}
\end{center}

\begin{center}
\scalebox{.95}{
\begin{tabular}{c c} 

\AxiomC{$\R, a \in D_{w}, w: A[a/x], \Gamma \Rightarrow \Delta$}
\RightLabel{$\existsl^{\dag_{3}}$}
\UnaryInfC{$\R, w : \exists x A, \Gamma \Rightarrow \Delta$}
\DisplayProof

&

\AxiomC{$\R, w \leq v, a \in D_{v}, v : A[a/x], w : \forall x A, \Gamma \Rightarrow \Delta$}
\RightLabel{$\alll$}
\UnaryInfC{$\R, w \leq v, a \in D_{v}, w : \forall x A, \Gamma \Rightarrow \Delta$}
\DisplayProof

\end{tabular}
}
\end{center}

\begin{center}
\begin{tabular}{c @{\hskip 1em} c}

\AxiomC{$\R, w \leq v, a \in D_{w}, a \in D_{v}, \Gamma \Rightarrow \Delta$}
\RightLabel{$\nd$}
\UnaryInfC{$\R, w \leq v, a \in D_{w}, \Gamma \Rightarrow \Delta$}
\DisplayProof

&

\AxiomC{$\R, w \leq v, a \in D_{v}, a \in D_{w}, \Gamma \Rightarrow \Delta$}
\RightLabel{$\cd$}
\UnaryInfC{$\R, w \leq v, a \in D_{v}, \Gamma \Rightarrow \Delta$}
\DisplayProof

\end{tabular}
\end{center}

\hrule
\caption{The labelled calculus $\lint$ for propositional intuitionistic logic consists of $\id$, $\imprr$, $\conrl$, $\conrr$, $\disrl$, $\disrr$, $\imprl$, $\refl$, $\trans$, and $\botr$ (see~\cite{DycNeg12}), and all rules give the calculus $\lintfocd$. The side condition $\dag_{1}$ states that the variable $v$ does not occur in the conclusion, $\dag_{2}$ states that neither $a$ nor $v$ occur in the conclusion, and $\dag_{3}$ states that $a$ does not occur in the conclusion. Labels and parameters restricted from occurring in the conclusion of an inference are called \emph{eigenvariables}.\protect\footnotemark}
\label{fig:propositional-calculus}
\end{figure}

\footnotetext{Note that $\id$ is an instance of $\idfo$; the same holds in the nested setting.}

We define a \emph{label substitution} $[w/v]$ on a labelled sequent in the usual way as the replacement of all labels $v$ occurring in the sequent with the label $w$. Similarly, we define a \emph{parameter substitution} $[a/b]$ on a labelled sequent as the replacement of all parameters $b$ occurring in the sequent with the parameter $a$.

\begin{theorem}
\label{thm:lint-properties} The calculi $\lint$ and $\lintfocd$ have the following properties:
\begin{itemize}

\item[$(i)$] 

\begin{itemize}

\item[(a)] For all $A \in \mathcal{L}$, $ \vdash_{\lint} \R,w \leq v, w : A, \Gamma \Rightarrow v : A, \Delta$;

\item[(b)] For all $A \in \mathcal{L}$, $ \vdash_{\lint} \R,w:A,\Gamma \Rightarrow \Delta, w :A$; 

\item[(c)] For all $A \in \mathcal{L_{Q}}$, $\vdash_{\lintfocd} \R,w \leq v, \vv{a} \in D_{w}, w : A(\vv{a}), \Gamma \Rightarrow v : A(\vv{a}), \Delta$; 

\item[(d)] For all $A \in \mathcal{L_{Q}}$, $\vdash_{\lintfocd} \R, \vv{a} \in D_{w}, w:A(\vv{a}),\Gamma \Rightarrow \Delta, w :A(\vv{a})$;

\end{itemize}

\item[$(ii)$] The $(lsub)$ and $(psub)$ rules are height-preserving $($i.e. `hp-'$)$ admissible;
\begin{center}
\begin{tabular}{c @{\hskip 1em} c}
\AxiomC{$\R,\Gamma \Rightarrow \Delta$}
\RightLabel{$(lsub)$}
\UnaryInfC{$\R[w/v],\Gamma[w/v] \Rightarrow \Delta[w/v]$}
\DisplayProof

&

\AxiomC{$\R,\Gamma \Rightarrow \Delta$}
\RightLabel{$(psub)$}
\UnaryInfC{$\R[a/b],\Gamma[a/b] \Rightarrow \Delta[a/b]$}
\DisplayProof
\end{tabular}
\end{center}


\item[$(iii)$] All rules are hp-invertible;

\item[$(iv)$] The $\wk$ and $\{(ctr_{R}),(ctr_{F_{l}}),(ctr_{F_{r}})\}$ rules (below) are hp-admissible;

\begin{center}
\begin{tabular}{c @{\hskip 2em} c} 
\AxiomC{$\R,\Gamma \Rightarrow \Delta$}
\RightLabel{$(wk)$}
\UnaryInfC{$\R',\R,\Gamma',\Gamma \Rightarrow \Delta',\Delta$}
\DisplayProof

&

\AxiomC{$\R,\R',\R',\Gamma \Rightarrow \Delta$}
\RightLabel{$(ctr_{R})$}
\UnaryInfC{$\R,\R',\Gamma \Rightarrow \Delta$}
\DisplayProof

\end{tabular}
\end{center}
\begin{center}
\begin{tabular}{c @{\hskip 2em} c} 

\AxiomC{$\R,\Gamma',\Gamma',\Gamma \Rightarrow \Delta$}
\RightLabel{$(ctr_{F_{l}})$}
\UnaryInfC{$\R,\Gamma',\Gamma \Rightarrow \Delta$}
\DisplayProof

&

\AxiomC{$\R,\Gamma \Rightarrow \Delta, \Delta', \Delta'$}
\RightLabel{$(ctr_{F_{r}})$}
\UnaryInfC{$\R,\Gamma \Rightarrow \Delta, \Delta'$}
\DisplayProof
\end{tabular}
\end{center}

\item[$(v)$] The $\cut$ rule (below) is admissible;

\begin{center}
\begin{tabular}{c}
\AxiomC{$\R,\Gamma \Rightarrow \Delta, w :A$}
\AxiomC{$\R,w :A,\Gamma \Rightarrow \Delta$}
\RightLabel{$\cut$}
\BinaryInfC{$\R,\Gamma \Rightarrow \Delta$}
\DisplayProof
\end{tabular}
\end{center}

\item[$(vi)$] $\lint$ $(\lintfocd)$ is sound and complete for $\int$ $(\intfocd$, resp.$)$.

\end{itemize}
\end{theorem}

\begin{proof} We refer the reader to~\cite{DycNeg12} for proofs of properties (i)--(vi) for $\lint$; note that hp-admissibility of $(psub)$ is trivial in the propositional setting since formulae do not contain parameters. The proofs of properties (i)--(vi) can be found in App.~\ref{app:proofs} for $\lintfocd$.
\qed
\end{proof}

\subsection{The Nested Calculi $\nint$ and $\nintfocd$}

We define a propositional (or, first-order) nested sequent $\Sigma$ to be a syntactic object defined via the following BNF grammars:
\begin{center}
\begin{tabular}{c @{\hskip 2em} c}
$X ::= A \ | \ X, X$ 

&

$\Sigma ::= X \far X \ | \ X \far X, [\Sigma], \ldots, [\Sigma]$
\end{tabular}
\end{center}
where $A$ is in the propositional language $\mathcal{L}$ (first-order language $\mathcal{L_{Q}}$, resp.). As in the previous section, we take the comma operator to be commutative and associative, allowing us to view (for example) syntactic entities $X$ as multisets.

In the first-order setting, we syntactically distinguish between bound variables and free variables in first-order formulae, using $\{x, y, z, \ldots\}$ for bound variables and replacing the occurrence of free variables in formulae with parameters $\{a, b, c, \ldots\}$. For example, we would use $p(a) \far p(b), [\bot \far \forall x q(x,b)]$ instead of the sequent $p(x) \far p(y), [\bot \far \forall x q(x,y)]$ in a nested derivation (where the free variable $x$ has been replaced by the parameter $a$ and $y$ has been replaced by $b$).




Nested sequents are often written as $\Sigma\{X \far Y, [\Sigma_{0}], \ldots, [\Sigma_{n}] \}$, which indicates that $X \far Y, [\Sigma_{0}], \ldots, [\Sigma_{n}]$ occurs at some depth in the nestings of the sequent $\Sigma$; e.g. if $\Sigma$ is taken to be $p(a) \far [\bot \far \forall x q(x,b), [ \far \top ]]$, then both $\Sigma\{\bot \far \forall x q(x,b)\}$ and $\Sigma\{ \far \top\}$ are correct representations of $\Sigma$ in our notation. The nested calculi are given in Fig.~\ref{fig:nested-calculus-propositional}.

\begin{figure}
\noindent\hrule

\begin{center}
\begin{tabular}{c c c}
\AxiomC{} \RightLabel{$\id$}
\UnaryInfC{$\Sigma \lcut X, p \far p, Y \rcut$}
\DisplayProof

&

\AxiomC{$\Sigma \lcut X, A,B \far Y \rcut $}
\RightLabel{$\conrl$}
\UnaryInfC{$\Sigma \lcut X, A \land B \far Y \rcut$}
\DisplayProof

&

\AxiomC{$\Sigma \lcut X \far A,B, Y \rcut $}
\RightLabel{$\disrr$}
\UnaryInfC{$\Sigma \lcut X \far A\lor B, Y \rcut$}
\DisplayProof
\end{tabular}
\end{center}

\begin{center}
\begin{tabular}{c c}
\AxiomC{$\Sigma \lcut X, A \far Y \rcut$}
\AxiomC{$\Sigma \lcut X, B \far Y \rcut$}
\RightLabel{$\disrl$}
\BinaryInfC{$\Sigma \lcut X, A \lor B \far Y \rcut$}
\DisplayProof

&

\AxiomC{$\Sigma \lcut X \far A, Y \rcut$}
\AxiomC{$\Sigma \lcut X \far B, Y \rcut$}
\RightLabel{$\conrr$}
\BinaryInfC{$\Sigma \lcut X \far A\land B, Y \rcut$}
\DisplayProof

\end{tabular}
\end{center}

\begin{center}
\begin{tabular}{c c c}
\AxiomC{$\Sigma \lcut X \far Y, [A \far ] \rcut$}
\RightLabel{$(\neg_{r})$}
\UnaryInfC{$\Sigma \lcut X \far Y, \neg A \rcut$}
\DisplayProof

&

\AxiomC{$\Sigma \lcut X \far A, Y \rcut$}
\RightLabel{$(\neg_{l})$}
\UnaryInfC{$\Sigma \lcut X, \neg A \far Y \rcut$}
\DisplayProof

&

\AxiomC{$\Sigma\{X \far Y, [X', A \far Y']\}$}
\RightLabel{$\lift$}
\UnaryInfC{$\Sigma\{X, A \far Y, [X' \far Y']\}$}
\DisplayProof
\end{tabular}
\end{center}

\begin{center}
\begin{tabular}{c c}
\AxiomC{$\Sigma \lcut X \far Y, [A \far B] \rcut$}
\RightLabel{$\imprr$}
\UnaryInfC{$\Sigma \lcut X \far A \imp B, Y \rcut$}
\DisplayProof

&

\AxiomC{$\Sigma \lcut X \far A, Y \rcut$}
\AxiomC{$\Sigma \lcut X, B \far Y \rcut$}
\RightLabel{$\imprl$}
\BinaryInfC{$\Sigma \lcut X, A \imp B \far Y \rcut$}
\DisplayProof
\end{tabular}
\end{center}

\begin{center}
\begin{tabular}{c c}
\AxiomC{}
\RightLabel{$(id_{q})$}
\UnaryInfC{$\Sigma \lcut X, p(\vv{a}) \far p(\vv{a}), Y \rcut$}
\DisplayProof

&

\AxiomC{$\Sigma \lcut X \far Y, A[a/x]  \rcut$}
\RightLabel{$(\exists_{r})$}
\UnaryInfC{$\Sigma \lcut X \far Y, \exists x A \rcut$}
\DisplayProof
\end{tabular}
\end{center}

\begin{center}
\begin{tabular}{c c c} 
 
\AxiomC{$\Sigma \lcut X \far Y, A[a/x]  \rcut$}
\RightLabel{$(\forall_{r})^{\dag}$}
\UnaryInfC{$\Sigma \lcut X \far Y, \forall x A \rcut$}
\DisplayProof

&

\AxiomC{$\Sigma \lcut X, A[a/x] \far Y  \rcut$}
\RightLabel{$(\forall_{l})$}
\UnaryInfC{$\Sigma \lcut X, \forall x A \far Y \rcut$}
\DisplayProof

&

\AxiomC{$\Sigma \lcut X, A[a/x] \far Y  \rcut$}
\RightLabel{$(\exists_{l})^{\dag}$}
\UnaryInfC{$\Sigma \lcut X, \exists x A \far Y \rcut$}
\DisplayProof
\end{tabular}
\end{center}

\hrule
\caption{Fitting's nested calculus $\nint$ for propositional intuitionistic logic consists of $\id$, $\conrl$, $\disrr$, $\disrl$, $\conrr$, $(\neg_{r})$, $(\neg_{l})$, $\lift$, $\imprr$, and $\imprl$. All rules taken together give the nested calculus $\nintfocd$~\cite{Fit14}. The side condition $\dag$ states that $a$ does not occur in the conclusion.}
\label{fig:nested-calculus-propositional}
\end{figure}

\begin{theorem}[Soundness and Completeness~\cite{Fit14}]\label{thm:nint-properties} The calculus $\nint$ $(\nintfocd)$ is sound and complete for $\int$ $(\intfocd$, resp.$)$.

\end{theorem}

\section{Translating Notation: Labelled and Nested}\label{section-notation}

It is instructive to observe that both nested and labelled sequents can be viewed as graphs (with the former restricted to trees and the latter more general). Graphs of sequents are significant for two reasons: the first (technical) reason is that graphs can be leveraged to switch from labelled to nested notation; thus, graphs will play a role in deriving our nested calculi from our labelled calculi. The second reason is that graphs offer insight into \emph{why} structural rule elimination yields nested systems, which will be discussed in the next section. 

It is straightforward to define the graph of each type of sequent. To do this, we first introduce a bit of notation and define the multiset $\Gamma \restriction w := \{A \ | \ w : A \in \Gamma\}$. For a labelled sequent $\Lambda = \R, \Gamma \Rightarrow \Delta$, the \emph{graph} $G(\Lambda)$ is the tuple $( V, E, \lambda )$, where (i) $V = \{w \ | \ \text{$w$ is a label in $\Lambda$.}\}$, (ii) $(w,v) \in E$ iff $w \leq v \in \R$, and
$$
(iii) \quad \lambda = \{(w,\Gamma' \Rightarrow \Delta') \ | \ \Gamma' = \Gamma \restriction w \text{, } \Delta' = \Delta \restriction w \text{, and } w \in V\}.$$
For a nested sequent, the graph is defined inductively on the structure of the nestings; we use strings $\sigma$ of natural numbers to denote vertices in the graph, similar to the prefixes used in prefixed tableaux~\cite{Fit72,Fit12,Fit14}.

\emph{Base case.} Let our nested sequent be of the form $X \far Y$ with $X$ and $Y$ multisets of formulae. Then, $G_{\sigma}(X \far Y) := (V_{\sigma},E_{\sigma},\lambda_{\sigma})$, where (i) $V_{\sigma} := \{\sigma\}$, (ii) $E_{\sigma} := \emptyset$, and (iii) $\lambda_{\sigma} := \{(\sigma, X \far Y)\}$.

\emph{Inductive step.} Suppose our nested sequent is of the form $X \far Y, [\Sigma_{0}], \ldots, [\Sigma_{n}]$. We assume that each $G_{\sigma i}(\Sigma_{i}) = ( V_{ \sigma i}, E_{\sigma i}, \lambda_{\sigma i} )$ (with $i \in \{0, \ldots, n\}$) is already defined, and define $G_{\sigma}(X \far Y, [\Sigma_{0}], \ldots, [\Sigma_{n}]) := ( V_{\sigma}, E_{\sigma}, \lambda_{\sigma} )$ as follows:
$$
(i) \quad V_{\sigma} := \{\sigma\} \cup \displaystyle{\bigcup_{0 \leq i \leq n} V_{\sigma i}} \qquad (ii) \quad E_{\sigma} := \{(\sigma,\sigma i) \ | \ 0 \leq i \leq n \} \cup \displaystyle{\bigcup_{0 \leq i \leq n} E_{\sigma i}}
$$
$$(iii) \quad \lambda_{\sigma} := \{(\sigma, X \far Y)\} \cup \displaystyle{\bigcup_{0 \leq i \leq n} \lambda_{\sigma i}}
$$



\begin{definition} Let $G_{0} = (V_{0},E_{0},\lambda_{0})$ and $G_{1} = (V_{1},E_{1},\lambda_{1})$ be two graphs. We define an \emph{isomorphism} $f : V_{0} \mapsto V_{1}$ between $G_{0}$ and $G_{1}$ to be a function such that: (i) $f$ is bijective, (ii) $(x,y) \in E_{0}$ iff $(fx,fy) \in E_{1}$, (iii) $\lambda_{0}(x) = \lambda_{1}(fx)$. We say $G_{0}$ and $G_{1}$ are \emph{isomorphic} iff there exists an isomorphism between them.

\end{definition}

Although the formal definitions above may appear somewhat cumbersome, the example below shows that transforming a sequent into its graph---or conversely, obtaining the sequent from its graph---is relatively straightforward. 

\begin{example}\label{ex:graphs-of-sequents} The nested sequent $\Sigma$ is given below with its corresponding graph $G_{0}(\Sigma)$ shown on the left, and the labelled sequent $\Lambda$ is given below with its corresponding graph $G(\Lambda)$ on the right. Regarding the labelled sequent, we assume that $\Gamma_{i}$ and $\Delta_{i}$ consist solely of formulae labelled with $w_{i}$ (for $i \in \{0,1,2,3\}$).
$$
\Sigma = X_{0} \far Y_{0}, [X_{1} \far Y_{1}, [X_{2} \far Y_{2}]], [X_{3} \far Y_{3}]
$$
\vspace*{0em}
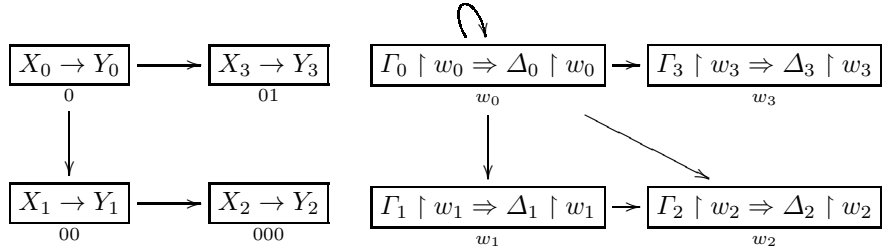
\begin{center}
\begin{tabular}{c @{\hskip 1em} c}
\xymatrix{
  \underset{0}{\boxed{X_{0} \far Y_{0}}}\ar@{->}[r]\ar@{>}[d] & \underset{01}{\boxed{X_{3} \far Y_{3}}}		\\
 	\underset{00}{\boxed{X_{1} \far Y_{1}}}\ar@{>}[r] & \underset{000}{\boxed{X_{2} \far Y_{2}}}   
}

&

\xymatrix@C=1em{
  \underset{w_{0}}{\boxed{\Gamma_{0} \restriction w_{0} \Rightarrow \Delta_{0} \restriction w_{0}}}\ar@{->}[r]\ar@{>}[d]\ar@(ul,u)\ar@{>}[dr] & \underset{w_{3}}{\boxed{\Gamma_{3} \restriction w_{3} \Rightarrow \Delta_{3} \restriction w_{3}}}		\\
 	\underset{w_{1}}{\boxed{\Gamma_{1} \restriction w_{1} \Rightarrow \Delta_{1} \restriction w_{1}}}\ar@{>}[r] & \underset{w_{2}}{\boxed{\Gamma_{2} \restriction w_{2} \Rightarrow \Delta_{2} \restriction w_{2}}}   
}
\end{tabular}
\end{center}
\begin{small}
$$
\Lambda = w_{0} \leq w_{0}, w_{0} \leq w_{1}, w_{1} \leq w_{2}, w_{0} \leq w_{2}, w_{0} \leq w_{3}, \Gamma_{0}, \Gamma_{1} , \Gamma_{2}, \Gamma_{3} \Rightarrow \Delta_{0}, \Delta_{1}, \Delta_{2}, \Delta_{3}   
$$
\end{small}
\end{example}

In the above example there is a loop from $w_{0}$ to itself in the graph of the labelled sequent; furthermore, there is an undirected cycle occurring between $w_{0}$, $w_{1}$, and $w_{2}$. As will be explained in the next section (specifically, Thm.~\ref{thm:treelike-derivations}), the $\refl$ and $\trans$ rules allow for such structures to appear in labelled derivations of theorems; however, the elimination of these rules in the labelled calculus has the effect that such structures \emph{can no longer} occur in the labelled derivation of a theorem. Consequently, it will be seen that eliminating such rules yields a labelled derivation where every sequent has a purely \emph{treelike} structure (see Def.~\ref{def:treelike}). This implies that each labelled sequent in the derivation has a graph isomorphic to the graph of a nested sequent. It is this idea which ultimately permits the extraction of our nested calculi from our labelled calculi.

\begin{definition}\label{def:treelike} Let $\Lambda$ be a labelled sequent and $G(\Lambda) = (V,E,\lambda)$. We say that $\Lambda$ is \emph{treelike} iff there exists a unique vertex $w \in V$, called the \emph{root}, such that there exists a unique path from $w$ to every other vertex $v \in V$.\footnote{Treelike sequents are equivalently characterized as sequents with graphs that are: (i) connected, (ii) acyclic, and (iii) contain no backwards branching.}

\end{definition}

If we take the graph of a treelike labelled sequent, then it can be viewed as the graph of a nested sequent, as the example below demonstrates.

\begin{example}\label{ex:switching-notation} The treelike labelled sequent $\Lambda'$ and its graph are given below. We assume that $\Gamma_{i}$ and $\Delta_{i}$ contain only formulae labelled with $w_{i}$ (for $i \in \{0,1,2,3\}$).
$$
\Lambda' = w_{0} \leq w_{1}, w_{1} \leq w_{2}, w_{0} \leq w_{3}, \Gamma_{0}, \Gamma_{1} , \Gamma_{2}, \Gamma_{3} \Rightarrow \Delta_{0}, \Delta_{1}, \Delta_{2}, \Delta_{3}   
$$
\begin{center}
\begin{tabular}{c} 

\xymatrix@C=1em{
\underset{w_{2}}{\boxed{\Gamma_{2}' \Rightarrow \Delta_{2}'}} & & \underset{w_{1}}{\boxed{\Gamma_{1}' \Rightarrow \Delta_{1}'}}\ar@{>}[ll] & & \underset{w_{0}}{\boxed{\Gamma_{0}' \Rightarrow \Delta_{0}'}}\ar@{->}[rr]\ar@{>}[ll] & & \underset{w_{3}}{\boxed{\Gamma_{3}' \Rightarrow \Delta_{3}'}}		
}
\end{tabular}
\end{center}
Also, we assume $\Gamma_{i}' = \Gamma_{i} \restriction w_{i} = X_{i}$ and $\Delta_{i} ' = \Delta_{i} \restriction w_{i} = Y_{i}$ (for $i \in \{0,1,2,3\}$). Therefore, the above graph is isomorphic to the graph of the nested sequent in Ex.~\ref{ex:graphs-of-sequents}, meaning that $\Lambda$ can be translated as that nested sequent.
\end{example}

\begin{definition}[The Translation $\switch$]\label{def:switch} Let $\Lambda$ be a treelike labelled sequent. We define $\switch(\Lambda)$ to be the nested sequent obtained from the graph $G(\Lambda)$.
\end{definition}

\section{Deriving $\nint$ from $\lint$}\label{section-3}


We begin by presenting two useful lemmata that will be referenced in the current and next section while deriving $\nint$ from $\lint$ and $\nintfocd$ from $\lintfocd$. All rules mentioned in the lemmata can be found in Fig.~\ref{fig:new-rules-focd} below. The proofs of both lemmata can be found in App.~\ref{app:proofs}.

\begin{lemma}
\label{lem:extended-lint-properties} The calculus $\lint + \{(id^{*}), (\neg_{l}), (\neg_{r}), (\imp^{*}_{l}), \lift\}$ and the calculus $\lintfocd + \{\idfonew, (\neg_{l}), (\neg_{r}), (\imp^{*}_{l}), (\forall_{l}^{*}), (\forall^{*}_{r}), (\exists_{r}^{*}), \lift\}$ have the following properties:
\begin{itemize}

\item[$(i)$] All sequents of the form $\R,w \leq v, \vv{a} \in D_{w}, w : A(\vv{a}), \Gamma \Rightarrow v : A(\vv{a}), \Delta$ and $\R, \vv{a} \in D_{w}, w:A(\vv{a}),\Gamma \Rightarrow \Delta, w :A(\vv{a})$ are derivable;\footnote{In the propositional setting, these sequents become $\R,w \leq v, w : A, \Gamma \Rightarrow v : A, \Delta$ and $\R, w:A,\Gamma \Rightarrow \Delta, w :A$, respectively.}







\item[(ii)] The rules $\{(lsub),(psub),\wk,(ctr_{R}),(ctr_{F_{r}})\}$ are hp-admissible;

\item[(iii)] With the exception of $\{\conrl,\existsl\}$, all rules are hp-invertible;

\item[(iv)] The rules $\{\conrl,\existsl\}$ are invertible;

\item[(v)] The rule $(ctr_{F_{l}})$ is admissible.

\end{itemize}
\end{lemma}


\begin{lemma}
\label{lm:structural-rules-permutation} (i) $\refl$ and $\trans$ can be permuted above each rule in the set $\{\botr, \conrl, \conrr, \disrl, \disrr, \imprr, (\neg_{l}), (\neg_{r}), (\exists_{l}), (\exists_{r}), (\forall_{r})\}$. (ii) $\nd$ and $\cd$ can be permuted above $\{\botr, \conrl, \conrr, \disrl, \disrr, \imprl, \imprr, (\neg_{l}), (\neg_{r}), (\exists_{l}), (\forall_{r})\}$.

\end{lemma}

\begin{proof} Claim (i) follows from the fact that none of the rules mentioned have active relational atoms of the form $w \leq u$ in the conclusion, and so, $\refl$ and $\trans$ may be freely permuted above each rule. Claim (ii) follows from the fact that none of the rules mentioned contain active domain atoms in the conclusion, allowing for $\nd$ and $\cd$ to be permuted above each rule.
\qed
\end{proof}

\begin{figure}
\noindent\hrule

\begin{center}
\resizebox{\columnwidth}{!}{
\begin{tabular}{c c}
\AxiomC{}
\RightLabel{$(id_{q}^{*})^{\dag_{1}}$}
\UnaryInfC{$\R, a_{0} \in D_{v_{0}}, \ldots, a_{n} \in D_{v_{n}}, w : p(\vv{a}), \Gamma \Rightarrow w : p(\vv{a}), \Delta$}
\DisplayProof

&

\AxiomC{}
\RightLabel{$\idnew$}
\UnaryInfC{$\R, w : p, \Gamma \Rightarrow \Delta, w : p$}
\DisplayProof
\end{tabular}
}
\end{center}

\begin{center}
\begin{tabular}{c @{\hskip 1em} c}

\AxiomC{$\R,w \leq v, v:A, \Gamma \Rightarrow \Delta$}
\RightLabel{$(\neg_{r})^{\dag_{2}}$}
\UnaryInfC{$\R, \Gamma \Rightarrow \Delta, w : \neg A$}
\DisplayProof

&

\AxiomC{$\R, w: \neg A, \Gamma \Rightarrow \Delta, w:A$}
\RightLabel{$(\neg_{l})$}
\UnaryInfC{$\R, w: \neg A, \Gamma \Rightarrow \Delta$}
\DisplayProof
\end{tabular}
\end{center}

\begin{center}
\begin{tabular}{c c}
\AxiomC{$\R, a \in D_{v}, w : A[a/x], w : \forall x A, \Gamma \Rightarrow \Delta$}
\RightLabel{$(\forall_{l}^{*})^{\dag_{3}}$}
\UnaryInfC{$\R, a \in D_{v}, w : \forall x A, \Gamma \Rightarrow \Delta$}
\DisplayProof
&
\AxiomC{$\R, a \in D_{w}, \Gamma \Rightarrow w : A[a/x], \Delta$}
\RightLabel{$(\forall^{*}_{r})^{\dag_{4}}$}
\UnaryInfC{$\R, \Gamma \Rightarrow w : \forall x A, \Delta$}
\DisplayProof
\end{tabular}
\end{center}

\begin{center}
\begin{tabular}{c c}
\AxiomC{$\R, w \leq u, w:A, u:A, \Gamma \Rightarrow \Delta$}
\RightLabel{$\lift$}
\UnaryInfC{$\R, w \leq u, w:A, \Gamma \Rightarrow \Delta$}
\DisplayProof

&

\AxiomC{$\R, a \in D_{v}, \Gamma \Rightarrow \Delta, w: A[a/x], w: \exists x A$}
\RightLabel{$(\exists_{r}^{*})^{\dag_{3}}$}
\UnaryInfC{$\R, a \in D_{v}, \Gamma \Rightarrow \Delta, w: \exists x A$}
\DisplayProof
\end{tabular}
\end{center}

\begin{center}
\AxiomC{$\R, w : A \imp B, \Gamma \Rightarrow \Delta, w : A$}
\AxiomC{$\R, w : A \imp B, w : B, \Gamma \Rightarrow \Delta$}
\RightLabel{$(\imp_{l}^{*})$}
\BinaryInfC{$\R, w : A \imp B, \Gamma \Rightarrow \Delta$}
\DisplayProof
\end{center}

\hrule
\caption{Rules used to derive $\nint$ and $\nintfocd$ from $\lint$ and $\lintfocd$, respectively. The side condition $\dag_{1}$ states that there exists a path of relational atoms (not necessarily directed) from $v_{i}$ to $w$ for each $i \in \{0, \ldots, n\}$ in $\R$; $\dag_{2}$ states that $v$ does not occur in the conclusion; $\dag_{3}$ stipulates that there exists a path of relational atoms (not necessarily directed) from $v$ to $w$ occurring in $\R$; $\dag_{4}$ states that $a$ does not occur in the conclusion.\protect\footnotemark}
\label{fig:new-rules-focd}
\end{figure}

\footnotetext{Let $u \sim v \in \{u \leq v, v \leq u\}$. A path of relational atoms (not necessarily directed) from a label $w$ to $v$ occurs in a sequent $\Lambda$ if and only if $w = v$, $w \sim v$, or there exist labels $z_{i}$ ($i \in \{0,\ldots,n\}$) such that $w \sim z_{0}, \ldots, z_{n} \sim v$ occurs in $\Lambda$.}


Deriving the calculus $\nint$ from $\lint$ depends on a crucial observation made in~\cite{CiaLyoRam18} concerning labelled derivations: \emph{rules such as $\refl$ and $\trans$ allow for theorems to be derived in proofs containing non-treelike labelled sequents}. To demonstrate this fact, observe the following derivation in $\lint$:
\begin{center}
\AxiomC{$w \leq v, v \leq v, v : p \Rightarrow v : p$}
\RightLabel{$\refl$}
\UnaryInfC{$w \leq v, v : p \Rightarrow v : p$}
\RightLabel{$\imprr$}
\UnaryInfC{$\Rightarrow w : p \imp p$}
\DisplayProof
\end{center}
The initial sequent is non-treelike due to the presence of the $v \leq v$ relational atom; however, the application of $\refl$ deletes this structure from the initial sequent and produces a treelike sequent as the conclusion.

In fact, it is true in general that every labelled derivation of a theorem (i.e., a derivation whose end sequent is of the form $\Rightarrow w : A$) can be partitioned into a top derivation consisting of non-treelike sequents, and a bottom derivation consisting of treelike sequents. Note that if a derivation ends with a sequent of the form $\Rightarrow w : A$, then the derivation must necessarily contain a bottom treelike fragment since $G(\Rightarrow w : A)$ is a tree. By contrast, the top non-treelike fragment of the derivation may be empty (e.g. the derivation of $\Rightarrow w : \bot \imp A$).

To demonstrate why the aforementioned partition always exists, suppose you are given a labelled derivation of a theorem $w : A$ and consider the derivation in a bottom-up manner. The graph of the end sequent $\Rightarrow w : A$ is evidently treelike by Def.~\ref{def:treelike}, and observe the each bottom-up application of a rule in $\lint$---with the exception of $\refl$ and $\trans$---will produce a treelike sequent (see Thm.~\ref{thm:treelike-derivations} for auxiliary details). If, however, at some point in the derivation $\refl$ or $\trans$ is applied, then all sequents above the inference will inherit the (un)directed cycle produced by the rule, thus producing the non-treelike fragment of the proof.

One can therefore imagine that permuting instances of the $\refl$ and $\trans$ rules upwards in a given derivation would potentially increase the bottom treelike fragment of the derivation and decrease the top non-treelike fragment. As it so happens, this intuition is correct so long as we choose \emph{adequate} rules---that bottom-up preserve the treelike structure of sequents---to replace certain instances of the $\refl$ and $\trans$ rules in a derivation, when necessary. We will first examine permuting instances of the $\refl$ rule, and motivate which adequate rules we ought to add to our calculus in order to achieve the complete elimination of $\refl$. After, we will turn our attention towards eliminating the $\trans$ rule, and conclude the section by leveraging our results to extract $\nint$.

Let us first observe an application of $\refl$ to an initial sequent obtained via the $\id$ rule. There are two possible cases to consider: either the relational atom principal in the initial sequent is active in the $\refl$ inference (shown below left), or it is not (shown below right):
\vspace*{-1em}
\begin{center}
\begin{tabular}{c c}
\AxiomC{}
\RightLabel{$\id$}
\UnaryInfC{$\R, w \leq w, w : p, \Gamma \Rightarrow \Delta, w : p$}
\RightLabel{$\refl$}
\UnaryInfC{$\R, w : p, \Gamma \Rightarrow \Delta, w : p$}
\DisplayProof

&

\AxiomC{}
\RightLabel{$\id$}
\UnaryInfC{$\R, u \leq u, w \leq v, w : p, \Gamma \Rightarrow \Delta, v : p$}
\RightLabel{$\refl$}
\UnaryInfC{$\R, w \leq v, w : p, \Gamma \Rightarrow \Delta, v : p$}
\DisplayProof
\end{tabular}
\end{center}
In the case shown above right, the end sequent is an instance of the $\id$ rule, regardless of if $u = w$, $u = v$, or $u$ is distinct from $w$ and $v$. The case shown above left, however, indicates that we ought to add the $\idnew$ rule (see Fig.~\ref{fig:new-rules-focd}) to our calculus if we aim to eliminate $\refl$ from any given derivation.

Additionally, notice that an application of $\refl$ to $\idnew$ or $\botr$ produces another instance of $\idnew$ or $\botr$, regardless of if $v = w$ or $v \neq w$.
\begin{center}
\begin{tabular}{c @{\hskip 2em} c}
\AxiomC{}
\RightLabel{$\idnew$}
\UnaryInfC{$\R, v \leq v, w : p, \Gamma \Rightarrow \Delta, w : p$}
\RightLabel{$\refl$}
\UnaryInfC{$\R, w : p, \Gamma \Rightarrow \Delta, w : p$}
\DisplayProof

&

\AxiomC{}
\RightLabel{$\botr$}
\UnaryInfC{$\R, v \leq v, w : \bot, \Gamma \Rightarrow \Delta$}
\RightLabel{$\refl$}
\UnaryInfC{$\R, w : \bot, \Gamma \Rightarrow \Delta$}
\DisplayProof
\end{tabular}
\end{center}
The above facts, coupled with Lem.~\ref{lm:structural-rules-permutation}, imply that any application of $\refl$ to an initial sequent, produces an initial sequent.

Concerning the remaining rules of $\lint$, we need only investigate the permutation of $\refl$ above the $\imprl$ rule, if we rely on Lem.~\ref{lm:structural-rules-permutation}. 
There are two cases: either the relational atom principal in the $\imprl$ inference is active in the $\refl$ inference, or it is not. The latter case is easily resolved, so we observe the former:
\vspace*{-2em}
\begin{center}
\scalebox{.95}{
\begin{tabular}{c}
\AxiomC{$\R, w \leq w, w : A \imp B, \Gamma \Rightarrow \Delta, w : A$}
\AxiomC{$\R, w \leq w, w : A \imp B, w : B, \Gamma \Rightarrow \Delta$}
\RightLabel{$\imprl$}
\BinaryInfC{$\R, w \leq w, w : A \imp B, \Gamma \Rightarrow \Delta$}
\RightLabel{$\refl$}
\UnaryInfC{$\R, w : A \imp B, \Gamma \Rightarrow \Delta$}
\DisplayProof
\end{tabular}
}
\end{center}
Applying $\refl$ to each premise of the $\imprl$ inference yields the following:
\begin{center}
\scalebox{.85}{
\begin{tabular}{c @{\hskip 1em} c}
\AxiomC{$\R, w \leq w, w : A \imp B, \Gamma \Rightarrow \Delta, w : A$}
\RightLabel{$\refl$}
\UnaryInfC{$\R, w : A \imp B, \Gamma \Rightarrow \Delta, w : A$}
\DisplayProof

&

\AxiomC{$\R, w \leq w, w : A \imp B, w : B, \Gamma \Rightarrow \Delta$}
\RightLabel{$\refl$}
\UnaryInfC{$\R, w : A \imp B, w : B, \Gamma \Rightarrow \Delta$}
\DisplayProof
\end{tabular}
}
\end{center}
The above observation suggests that we ought to add the $(\imp_{l}^{*})$ rule (see Fig.~\ref{fig:new-rules-focd}) to our calculus if we wish to permute $\refl$ above the $\imprl$ rule; a single application of the $(\imp_{l}^{*})$ rule to the end sequents above gives the desired conclusion.

With the $(\imp_{l}^{*})$ rule added to our calculus, we may freely permute the $\refl$ rule above any $\imprl$ inference. Still, we must confirm that the $\refl$ rule is permutable with the newly introduced $(\imp_{l}^{*})$ rule, but this is easily verifiable.

On the basis of our investigation, we may conclude the following lemma:

\begin{lemma}\label{lm:refl-admiss} The $\refl$ rule is eliminable in $\lint + \{(id^{*}),(\imp^{*}_{l})\} - \trans$.

\end{lemma}

Let us turn our attention towards eliminating the $\trans$ rule from a labelled derivation. Since our aim is to eliminate \emph{both} $\refl$ and $\trans$ from any derivation, we assume that the rules $\{(id^{*}),(\imp^{*}_{l})\}$ have been added to our calculus.

It is rather simple to verify that $\trans$ permutes with $\botr$ and $(id^{*})$, so we only consider the $\id$ case. As with the $\refl$ rule, there are two cases to consider when permuting $\trans$ above an $\id$ inference: either, the active formula of $\trans$ is principal in $\id$, or it is not. In the latter case, the result of the $\trans$ rule is an initial sequent, implying that the $\trans$ rule may be eliminated from the derivation. The former case proves trickier and is explicitly given below:
\begin{center}
\begin{tabular}{c}
\AxiomC{}
\RightLabel{$\id$}
\UnaryInfC{$\R, w \leq u, u \leq v, w \leq v, w : p, \Gamma \Rightarrow \Delta, v : p$}
\RightLabel{$\trans$}
\UnaryInfC{$\R, w \leq u, u \leq v, w : p, \Gamma \Rightarrow \Delta, v : p$}
\DisplayProof
\end{tabular}
\end{center}
Observe that the end sequent is not an initial sequent as it is not obtainable from an $\id$, $(id^{*})$, or $\botr$ rule. The issue is solved by considering the $\lift$ rule (see Fig.~\ref{fig:new-rules-focd}), 
which allows us to obtain the desired end sequent without the use of $\trans$, as the following derivation demonstrates:
\begin{center}
\AxiomC{}
\RightLabel{$(id^{*})$}
\UnaryInfC{$\R, w \leq u, u \leq v, w:p, u:p, v:p, \Gamma \Rightarrow v : p, \Delta$}
\RightLabel{$\lift$}
\UnaryInfC{$\R, w \leq u, u \leq v, w:p, u:p, \Gamma \Rightarrow v : p, \Delta$}
\RightLabel{$\lift$}
\UnaryInfC{$\R, w \leq u, u \leq v, w:p, \Gamma \Rightarrow v : p, \Delta$}
\DisplayProof
\end{center}
Thus, the addition of $\lift$ to our calculus resolves the issue of permuting $\trans$ above any initial sequent. Nevertheless, by Lem.~\ref{lm:structural-rules-permutation}, we still need to consider the permutation of $\trans$ above the $\imprl$, $(\imp_{l}^{*})$, and $\lift$ rules. The $\trans$ rule and $(\imp_{l}^{*})$ are freely permutable due to the fact that $\trans$ solely affects relational atoms, and $(\imp_{l}^{*})$ solely affects labelled formulae. Also, the following lemma entails that we may omit analyzing the permutation of $\trans$ above the $\imprl$ rule.

\begin{lemma}\label{lm:implies-left-deriv} The rule $\imprl$ is admissible in $\lint + \{(id^{*}), (\imp^{*}_{l}), \lift\}$.

\end{lemma}

\begin{proof} We derive the rule as shown below:

\begin{center}
\resizebox{\columnwidth}{!}{
\begin{tabular}{c}
\AxiomC{$\R,x \leq y, x:A \imp B, \Gamma \Rightarrow \Delta, y:A$}
\RightLabel{Lem.~\ref{lem:extended-lint-properties}}
\dashedLine
\UnaryInfC{$\R,x \leq y, x:A \imp B,y:A \imp B, \Gamma \Rightarrow \Delta, y:A$}
\AxiomC{$\R,x \leq y, x:A \imp B, y:B, \Gamma \Rightarrow \Delta$}
\RightLabel{Lem.~\ref{lem:extended-lint-properties}}
\dashedLine
\UnaryInfC{$\R,x \leq y, x:A \imp B,y:A \imp B, y:B, \Gamma \Rightarrow \Delta$}
\RightLabel{$(\imp^{*}_{l})$}
\BinaryInfC{$\R, x \leq y, x:A \imp B, y:A \imp B, \Gamma \Rightarrow \Delta$}
\RightLabel{$\lift$}
\UnaryInfC{$\R, x \leq y, x:A \imp B, \Gamma \Rightarrow \Delta$}
\DisplayProof 
\end{tabular}
}
\end{center}
\qed
\end{proof}

Last, the $\trans$ rule is permutable with the $\lift$ rule. In the case where the principal relational atom of $\lift$ is \emph{not} active in the ensuing $\trans$ application, the two rules freely permute. The alternative case is resolved as shown below:
\begin{center}
\resizebox{\columnwidth}{!}{
\begin{tabular}{c c}
\AxiomC{$\R, w \leq u, u \leq v, w \leq v, w: A, v : A, \Gamma \Rightarrow \Delta$}
\RightLabel{$\lift$}
\UnaryInfC{$\R, w \leq u, u \leq v, w \leq v, w: A, \Gamma \Rightarrow \Delta$}
\RightLabel{$\trans$}
\UnaryInfC{$\R, w \leq u, u \leq v, w: A, \Gamma \Rightarrow \Delta$}
\DisplayProof
&
\AxiomC{$\R, w \leq u, u \leq v, w \leq v, w: A,v : A, \Gamma \Rightarrow \Delta$}
\RightLabel{$\trans$}
\UnaryInfC{$\R, w \leq u, u \leq v, w: A,v : A, \Gamma \Rightarrow \Delta$}
\RightLabel{Lem.~\ref{lem:extended-lint-properties}}
\dashedLine
\UnaryInfC{$\R, w \leq u, u \leq v, w: A, u:A, v : A, \Gamma \Rightarrow \Delta$}
\RightLabel{$\lift$}
\UnaryInfC{$\R, w \leq u, u \leq v, w: A, u:A, \Gamma \Rightarrow \Delta$}
\RightLabel{$\lift$}
\UnaryInfC{$\R, w \leq u, u \leq v, w: A, \Gamma \Rightarrow \Delta$}
\DisplayProof
\end{tabular}
}
\end{center}



Hence, we obtain the following:

\begin{lemma}\label{lm:trans-admiss} The $\trans$ rule is eliminable in $\lint + \{(id^{*}),(\imp^{*}_{l}),\lift\} - \refl$.

\end{lemma}

Enough groundwork has been laid to state our main lemma:

\begin{lemma}\label{lm:refl-trans-admiss} The $\refl$ and $\trans$ rules are admissible in the calculus $\lint + \{(id^{*}),(\imp^{*}_{l}),\lift\}$.

\end{lemma}

\begin{proof} Suppose we are given a proof $\Pi$ in $\lint + \{(id^{*}),(\imp^{*}_{l}),\lift\}$, and consider the topmost occurrence of either $\refl$ or $\trans$. If we can show that the $\refl$ rule permutes above the $\lift$ rule, then we can evoke Lem.~\ref{lm:refl-admiss} and Lem.~\ref{lm:trans-admiss} to conclude that each topmost occurrence of $\refl$ and $\trans$ can be eliminated from $\Pi$ in succession. This yields a $\refl$ and $\trans$ free proof of the end sequent and establishes the claim. Thus, we prove that the $\refl$ rule permutes above the $\lift$ rule.

In the case where the relational atom active in $\refl$ is not principal in the $\lift$ inference, the two rules may be permuted; the alternative case is resolved as shown below:
\begin{center}
\begin{tabular}{c @{\hskip 1em} c}

\AxiomC{$\R, w \leq w, w: A, w : A, \Gamma \Rightarrow \Delta$}
\RightLabel{$\lift$}
\UnaryInfC{$\R, w \leq w, w:A, \Gamma \Rightarrow \Delta$}
\RightLabel{$\refl$}
\UnaryInfC{$\R, w:A, \Gamma \Rightarrow \Delta$}
\DisplayProof

&

\AxiomC{}
\RightLabel{IH}
\dashedLine
\UnaryInfC{$\R, w: A, w : A, \Gamma \Rightarrow \Delta$}
\RightLabel{Lem.~\ref{lem:extended-lint-properties}-(v)}
\dashedLine
\UnaryInfC{$\R, w: A, \Gamma \Rightarrow \Delta$}
\DisplayProof

\end{tabular}
\end{center}
\qed
\end{proof}

The addition of the rules $\{\idnew, (\imp_{l}^{*}), \lift\}$ to our calculus and the above admissibility results demonstrate that we are readily advancing towards our goal of deriving $\nint$. Howbeit, 
our labelled calculus is still distinct since it makes use of the logical signature $\{\bot, \land, \lor, \imp\}$, whereas $\nint$ uses the signature $\{\neg, \land,\lor, \imp\}$. Therefore, we need to show that $\botr$ (we define $\bot := p \land \neg p$) is admissible in the presence of (labelled versions of) the $(\neg_{r})$ and $(\neg_{l})$ rules (see Fig.~\ref{fig:new-rules-focd}). 
This admissibility result is explained in the main theorem below.



\begin{theorem}
\label{thm:admiss-all-rules-propositional}
The rules $\{\id, \botr, \imprl, \refl, \trans\}$ are admissible in the calculus $\lint + \{\idnew, (\neg_{l}), (\neg_{r}), (\imp_{l}^{*}), \lift\}$.
\end{theorem}

\begin{proof} Follows from Lem.~\ref{lm:structural-rules-permutation}, Lem.~\ref{lm:refl-trans-admiss}, the fact that $\id$ is derivable using $\idnew$ and $\lift$, the fact that $\botr$ is derivable from $(\neg_{l})$ and $\conrl$, and admissibility of $\imprl$ is shown as in Lem.~\ref{lm:implies-left-deriv}.
\qed
\end{proof}

\begin{theorem}
\label{thm:treelike-derivations}
Every derivation $\Pi$ in $\lint + \{\idnew, (\neg_{l}), (\neg_{r}), (\imp_{l}^{*}), \lift\} - \{\id, \botr, \imprl, \refl, \trans\}$ of a labelled formula $w:A$ contains solely treelike sequents with $w$ the root of each sequent in the derivation.
\end{theorem}

\begin{proof} To prove the claim we have to show that the graph of every sequent is (i) connected, (ii) free of directed cycles, and (iii) free of backwards branching. (NB. properties (i)--(iii) are equivalent to being treelike.) We assume that we are given a derivation $\Pi$ with end sequent $\Rightarrow w:A$ and argue that every sequent in $\Pi$ has properties (i)--(iii):

(i) Assume there exists a sequent $\Lambda$ in $\Pi$ whose graph $G(\Lambda)$ is disconnected. Then, there exist at least two distinct regions in $G(\Lambda)$ such that there does not exist an edge from a node $v$ in one region to a node $u$ in the other region. In other words, $\Lambda$ does not contain a relational atom $v \leq u$ for some $v$ in one region and some $u$ in the other region. If one observes the rules of our calculus, they will find that all rules either preserve the relational atoms $\R$ of a sequent or decrease it by one relation atom (as in the case of the $(\neg_{r})$ and $\imprr$ rules). Hence, the end sequent $\Rightarrow w : A$ will have a disconnected graph since the property will be preserved downwards, but this is a contradiction.

(ii) Assume that a sequent occurs in $\Pi$ containing a relational cycle $u \leq v_{1}, \ldots, v_{n} \leq u$ (for $n \in \mathbb{N}$). 
Observe that the $(\neg_{r})$ and $\imprr$ rules are never applicable to any of the relational atoms in the cycle, since no label occurring in a relational cycle is an eigenvariable. This implies that the relational cycle will be preserved downwards into the end sequent $\Rightarrow w:A$ due to the fact that the $(\neg_{r})$ and $\imprr$ rules are the only rules that delete relational atoms, giving a contradiction.

(iii) Assume that a sequent occurs in $\Pi$ with backwards branching it is graph, i.e. it contains relational atoms of the form $v \leq z, u \leq z$. By reasoning similar to case (ii), we obtain a contradiction.

To see that $w$ is the root of each treelike sequent in $\Pi$, observe that applying inference rules from the calculus bottom-up to $w : A$ either preserve relational structure, or add forward relational structure (e.g. $\imprr$ and $(\neg_{r})$), thus constructing a tree emanating from $w$.
\qed
\end{proof}


We refer to the labelled calculus $\lint + \{\idnew, (\neg_{l}), (\neg_{r}), (\imp_{l}^{*}), \lift\} - \{\id, \botr, \imprl, \refl, \trans\}$ (restricted to the use of treelike sequents) under the $\switch$ translation as $\nint^{*}$. Up to copies of principal formulae in the premise(s) of some rules, $\nint^{*}$ is identical to the calculus $\nint$ (cf.~App.~\ref{app:new-nested-calculi}).






\section{Deriving $\nintfocd$ from $\lintfocd$}\label{section-4}


To simplify notation, $\lintfocd + \{\idfonew, (\neg_{l}), (\neg_{r}), (\imp^{*}_{l}), (\forall^{*}_{r}), (\forall_{l}^{*}), (\exists_{r}^{*}), \lift\}$ will be referred to as $\intfocdl$. 
We begin the section by showing two lemmata that confirm the admissibility of structural rules in $\intfocdl$, and permit the extraction of $\nintfocd$ from $\lintfocd$. After, we list a number of significant proof-theoretic properties inherited by our nested calculi through the extraction process.


\begin{lemma}
\label{lem:trans-refl-fo-elim}
The $\refl$ and $\trans$ rules are admissible in the calculus $\intfocdl$.
\end{lemma}

\begin{proof} We consider the topmost occurrence of a $\refl$ or $\trans$ inference and eliminate each topmost occurrence in succession until we obtain a proof free of $\refl$ and $\trans$ inferences. By Lem.~\ref{lm:structural-rules-permutation}, we need only show that $\refl$ and $\trans$ permute above rules $\{\idfo, (id^{*}_{q}), \imprl, (\imp_{l}^{*}), \alll, \allnewl, (\forall_{r}^{*}), (\exists^{*}_{r}), \lift, \nd, \cd\}$. 
The cases of permuting $\refl$ and $\trans$ above $(\imp_{l}^{*})$ and $\lift$ are similar to Lem.~\ref{lm:refl-admiss}, \ref{lm:trans-admiss}, and Thm.~\ref{thm:admiss-all-rules-propositional}. Also, $\imprl$ is admissible in the presence of $(\imp^{*}_{l})$ (similar to Lem.~\ref{lm:implies-left-deriv}), so the case may be omitted. Permuting $\refl$ and $\trans$ above $\idfonew$ and $(\forall_{r}^{*})$ is straightforward, so we exclude the cases. Hence, we focus only on the non-trivial cases involving the $\idfo$, $\alll$, $(\forall_{l}^{*})$, $(\exists^{*}_{r})$, $\nd$, and $\cd$ rules. We prove the elimination of $\refl$ and refer to reader to App.~\ref{app:proofs} for the proof of $\trans$ elimination. 
\begin{center}
\resizebox{\columnwidth}{!}{
\begin{tabular}{c c}
\AxiomC{}
\RightLabel{$\idfo$}
\UnaryInfC{$\R, w \leq w, \vv{a} \in D_{w}, w : p(\vv{a}), \Gamma \Rightarrow w : p(\vv{a}), \Delta$}
\RightLabel{$\refl$}
\UnaryInfC{$\R, \vv{a} \in D_{w}, w : p(\vv{a}), \Gamma \Rightarrow w : p(\vv{a}), \Delta$}
\DisplayProof
&
\AxiomC{}
\RightLabel{$\idfonew$}
\UnaryInfC{$\R, \vv{a} \in D_{w}, w : p(\vv{a}), \Gamma \Rightarrow w : p(\vv{a}), \Delta$}
\DisplayProof
\end{tabular}
}
\end{center}

\begin{center}
\resizebox{\columnwidth}{!}{
\begin{tabular}{c c}
\AxiomC{$\R, w \leq w, a \in D_{w}, w : A[a/x], w : \forall x A, \Gamma \Rightarrow \Delta$}
\RightLabel{$\alll$}
\UnaryInfC{$\R, w \leq w, a \in D_{w}, w : \forall x A, \Gamma \Rightarrow \Delta$}
\RightLabel{$\refl$}
\UnaryInfC{$\R, a \in D_{w}, w : \forall x A, \Gamma \Rightarrow \Delta$}
\DisplayProof
&
\AxiomC{}
\RightLabel{IH}
\dashedLine
\UnaryInfC{$\R, a \in D_{w}, w : A[a/x], w : \forall x A, \Gamma \Rightarrow \Delta$}
\RightLabel{$(\forall^{*}_{l})$}
\UnaryInfC{$\R, a \in D_{w}, w : \forall x A, \Gamma \Rightarrow \Delta$}
\DisplayProof
\end{tabular}
}
\end{center}

\begin{center}
\resizebox{\columnwidth}{!}{
\begin{tabular}{c c}
\AxiomC{$\R, u \leq u, a \in D_{v}, w : A[a/x], w : \forall x A, \Gamma \Rightarrow \Delta$}
\RightLabel{$(\forall^{*}_{l})$}
\UnaryInfC{$\R, u \leq u, a \in D_{v}, w : \forall x A, \Gamma \Rightarrow \Delta$}
\RightLabel{$\refl$}
\UnaryInfC{$\R, a \in D_{v}, w : \forall x A, \Gamma \Rightarrow \Delta$}
\DisplayProof
&
\AxiomC{}
\RightLabel{IH}
\dashedLine
\UnaryInfC{$\R, a \in D_{v}, w : A[a/x], w : \forall x A, \Gamma \Rightarrow \Delta$}
\RightLabel{$(\forall^{*}_{l})$}
\UnaryInfC{$\R, a \in D_{v}, w : \forall x A, \Gamma \Rightarrow \Delta$}
\DisplayProof
\end{tabular}
}
\end{center}

\begin{center}
\resizebox{\columnwidth}{!}{
\begin{tabular}{c c}
\AxiomC{$\R, u \leq u, a \in D_{v}, \Gamma \Rightarrow w : A[a/x], w : \exists x A, \Delta$}
\RightLabel{$(\exists^{*}_{r})$}
\UnaryInfC{$\R, u \leq u, a \in D_{v}, \Gamma \Rightarrow w : \exists x A, \Delta$}
\RightLabel{$\refl$}
\UnaryInfC{$\R, a \in D_{v}, \Gamma \Rightarrow w : \exists x A, \Delta$}
\DisplayProof
&
\AxiomC{}
\RightLabel{IH}
\dashedLine
\UnaryInfC{$\R, a \in D_{v}, \Gamma \Rightarrow w : A[a/x], w : \exists x A, \Delta$}
\RightLabel{$(\exists^{*}_{l})$}
\UnaryInfC{$\R, a \in D_{v}, \Gamma \Rightarrow w : \exists x A, \Delta$}
\DisplayProof
\end{tabular}
}
\end{center}

\begin{center}
\resizebox{\columnwidth}{!}{
\begin{tabular}{c c}
\AxiomC{$\R, w \leq w, a \in D_{w}, a \in D_{w}, \Gamma \Rightarrow \Delta$}
\RightLabel{$\nd$}
\UnaryInfC{$\R, w \leq w, a \in D_{w}, \Gamma \Rightarrow \Delta$}
\RightLabel{$\refl$}
\UnaryInfC{$\R, a \in D_{w}, \Gamma \Rightarrow \Delta$}
\DisplayProof

&

\AxiomC{$\R, w \leq w, a \in D_{w}, a \in D_{w}, \Gamma \Rightarrow \Delta$}
\RightLabel{Lem.~\ref{lem:extended-lint-properties}-(iv)}
\dashedLine
\UnaryInfC{$\R, w \leq w, a \in D_{w}, \Gamma \Rightarrow \Delta$}
\RightLabel{IH}
\dashedLine
\UnaryInfC{$\R, a \in D_{w}, \Gamma \Rightarrow \Delta$}
\DisplayProof
\end{tabular}
}
\end{center}

\begin{center}
\resizebox{\columnwidth}{!}{
\begin{tabular}{c c}
\AxiomC{$\R, w \leq w, a \in D_{w}, a \in D_{w}, \Gamma \Rightarrow \Delta$}
\RightLabel{$\cd$}
\UnaryInfC{$\R, w \leq w, a \in D_{w}, \Gamma \Rightarrow \Delta$}
\RightLabel{$\refl$}
\UnaryInfC{$\R, a \in D_{w}, \Gamma \Rightarrow \Delta$}
\DisplayProof

&

\AxiomC{$\R, w \leq w, a \in D_{w}, a \in D_{w}, \Gamma \Rightarrow \Delta$}
\RightLabel{Lem.~\ref{lem:extended-lint-properties}-(iv)}
\dashedLine
\UnaryInfC{$\R, w \leq w, a \in D_{w}, \Gamma \Rightarrow \Delta$}
\RightLabel{IH}
\dashedLine
\UnaryInfC{$\R, a \in D_{w}, \Gamma \Rightarrow \Delta$}
\DisplayProof
\end{tabular}
}
\end{center}

In the $(\forall^{*}_{l})$ and $(\exists^{*}_{r})$ cases, observe that the side condition continues to hold after IH is applied. If the path from $w$ to $v$ does not go through $u$, then the side
condition trivially holds, and if it does go through $u$, then there
must exist relational atoms in $\R$ occurring along the path from $w$
to $v$, which continue to be present after the evocation of IH.
\qed
\end{proof}

\begin{lemma}
\label{lm:nd-cd-admiss}
The rules $\nd$ and $\cd$ are admissible in the calculus $\intfocdl - \{\refl,\trans\}$.
\end{lemma}

\begin{proof} The result is shown by induction on the height of the given derivation by permuting all instances of $\nd$ and $\cd$ upwards until all such instances are removed from the derivation. See App.~\ref{app:proofs} for details.
\qed
\end{proof}

\begin{theorem}
\label{thm:focd-admissible-rules-all}
The rules $\{\idfo,\botr,\imprl,\alll,\allr,\existsr,\refl,\trans,\nd,\cd\}$ are admissible in $\intfocdl$.
\end{theorem}

\begin{proof} Admissibility of $\idfo$, $\imprl$, and $\botr$ is shown similarly to Lem.~\ref{lm:implies-left-deriv} and Thm.~\ref{thm:admiss-all-rules-propositional}. Also, the rule $\existsr$ is an instance of $(\exists_{r}^{*})$, and the admissibility of $\allr$ and $\alll$ are witnessed by the derivations below:
\begin{center}
\resizebox{\columnwidth}{!}{
\begin{tabular}{c c}
\AxiomC{$\R, w \leq v, a \in D_{v}, \Gamma \Rightarrow v : A[a/x], \Delta$}
\RightLabel{Lem.~\ref{lem:extended-lint-properties}}
\dashedLine
\UnaryInfC{$\R, w \leq w, a \in D_{w}, \Gamma \Rightarrow w : A[a/x], \Delta$}
\RightLabel{Lem.~\ref{lem:trans-refl-fo-elim}}
\dashedLine
\UnaryInfC{$\R, a \in D_{w}, \Gamma \Rightarrow w : A[a/x], \Delta$}
\RightLabel{$(\forall_{r}^{*})$}
\UnaryInfC{$\R, \Gamma \Rightarrow w : \forall x A, \Delta$}
\DisplayProof

&

\AxiomC{$\R, w \leq v, a \in D_{v}, v : A[a/x], w : \forall x A, \Gamma \Rightarrow \Delta$}
\RightLabel{Lem.~\ref{lem:extended-lint-properties}}
\dashedLine
\UnaryInfC{$\R, w \leq w, a \in D_{w}, w : A[a/x], w : \forall x A, \Gamma \Rightarrow \Delta$}
\RightLabel{Lem.~\ref{lem:trans-refl-fo-elim}}
\dashedLine
\UnaryInfC{$\R, a \in D_{w}, w : A[a/x], w : \forall x A, \Gamma \Rightarrow \Delta$}
\RightLabel{$(\forall_{l}^{*})$}
\UnaryInfC{$\R, a \in D_{w}, w : \forall x A, \Gamma \Rightarrow \Delta$}
\DisplayProof
\end{tabular}
}
\end{center}
Hence, our result follows by Lem.~\ref{lem:trans-refl-fo-elim} and Lem.~\ref{lm:nd-cd-admiss}.
\qed
\end{proof}

\begin{theorem}
\label{thm:treelike-derivations-QC}
Every derivation $\Pi$ of a labelled formula $w : A$ in the calculus $\intfocdl - \{\idfo,\botr,\imprl,\alll,\allr,\existsr,\refl,\trans,\nd,\cd\}$ contains solely treelike sequents with $w$ the root of each sequent in the derivation.
\end{theorem}

\begin{proof}
Similar to Thm.~\ref{thm:treelike-derivations}.
\qed
\end{proof}

We refer to $\intfocdl-\{\idfo, \imprl, \alll, \allr, \existsr, \refl,\trans,\nd,\cd\}$ (restricted to using treelike sequents) under the $\switch$ translation as $\nintfocd^{*}$. It is crucial to point out that by the definition of $\switch$ (Def.~\ref{def:switch}) and the definition of the graph of a labelled sequent, domain atoms $a \in D_{w}$ are omitted when translating from labelled to nested (and the $\{\idfonew,(\forall_{l}^{*}),(\exists_{r}^{*})\}$ side conditions become unnecessary). Hence, 
up to copies of principal formulae in the premises of some rules, $\nintfocd^{*}$ is identical to $\nintfocd$ (cf. App.~\ref{app:new-nested-calculi}). In fact, through additional work, one can eliminate such copies of principal formulae, begetting the complete extraction of $\nint$ and $\nintfocd$ from $\nint^{*}$ and $\nintfocd^{*}$ (and therefore, from $\lint$ and $\lintfocd$).


An interesting consequence of our work is that $\nint^{*}$ and $\nintfocd^{*}$ inherited favorable proof-theoretic properties as a consequence of their extraction. Such properties are listed in Cor.~\ref{cor:nint-nintfocd-inherit-proeprties} below with admissible rules found in Fig.~\ref{fig:admiss-rules-nint}.

\begin{figure}
\noindent\hrule

\begin{center}
\begin{tabular}{c @{\hskip 1em} c @{\hskip 1em} c} 
\AxiomC{$\Sigma\{X \far Y, [\Sigma'], [\Sigma']\}$}
\RightLabel{$(ctr_{1})$}
\UnaryInfC{$\Sigma\{X \far Y, [\Sigma']\}$}
\DisplayProof

&

\AxiomC{$\Sigma\{X,X \far Y\}$}
\RightLabel{$(ctr_{2})$}
\UnaryInfC{$\Sigma\{X \far Y\}$}
\DisplayProof

&

\AxiomC{$\Sigma\{X \far Y,Y\}$}
\RightLabel{$(ctr_{3})$}
\UnaryInfC{$\Sigma\{X \far Y\}$}
\DisplayProof


\end{tabular}
\end{center}

\begin{center}
\begin{tabular}{c @{\hskip 1em} c @{\hskip 1em} c}
\AxiomC{$\Sigma\{X \far Y\}$}
\RightLabel{$(wk_{1})$}
\UnaryInfC{$\Sigma\{X \far Y, [\Sigma']\}$}
\DisplayProof

&

\AxiomC{$\Sigma\{X \far Y\}$}
\RightLabel{$(wk_{2})$}
\UnaryInfC{$\Sigma\{X \far Y,Z\}$}
\DisplayProof

&

\AxiomC{$\Sigma\{X \far Y\}$}
\RightLabel{$(wk_{3})$}
\UnaryInfC{$\Sigma\{X,Z \far Y\}$}
\DisplayProof
\end{tabular}
\end{center}

\hrule
\caption{Examples of admissible structural rules in $\nint^{*}$ and $\nintfocd^{*}$.}
\label{fig:admiss-rules-nint}
\end{figure}

\begin{corollary}
\label{cor:nint-nintfocd-inherit-proeprties}
The calculi $\nint^{*}$ and $\nintfocd^{*}$ have inherited: (i) cut-free completeness, (ii) invertibility of all rules, and (iii) admissibility of all rules in Fig.~\ref{fig:admiss-rules-nint}.
\end{corollary}

\begin{proof} All properties follow from Lem.~\ref{lem:extended-lint-properties} and Thm.'s~\ref{thm:lint-properties}, \ref{thm:admiss-all-rules-propositional}, \ref{thm:treelike-derivations}, \ref{thm:focd-admissible-rules-all}, and \ref{thm:treelike-derivations-QC}. The admissibility of $(ctr_{1})$ follows from the admissibility of the rules in the set $\{(lsub),(ctr_{R}),(ctr_{F_{l}}),(ctr_{F_{r}})\}$, the admissibility of $(ctr_{2})$ follows from the admissibility of $(ctr_{F_{l}})$, the admissibility of $(ctr_{3})$ follows from the admissibility of $(ctr_{F_{r}})$, and the admissibility of $\{(wk_{1}), (wk_{2}), (wk_{3})\}$ follows from the admissibility of $\wk$ in the labelled variants of $\nint^{*}$ and $\nintfocd^{*}$.
\qed
\end{proof}

\section{Conclusion}\label{conclusion}

In this paper, we showed how to extract Fitting's nested calculi (up to copies of principal formulae in premises) from the labelled calculi $\lint$ and $\lintfocd$. The extraction is obtained via the elimination of structural rules and through the addition of special rules to $\lint$ and $\lintfocd$, necessitating the use of only treelike sequents in proofs of theorems. Consequently, the extraction of the nested calculi from the labelled calculi demonstrated that the former inherited favorable proof-theoretic properties from the latter (cut-free completeness, invertibility of rules, etc.). 


Regarding future work, the author aims to 
investigate 
modal and intermediate logics that allow for the extraction of cut-free nested calculi from their labelled calculi, as well as provide new nested calculi for logics lacking one. 
These results could also prove beneficial in the explication of a general methodology for obtaining nested calculi well-suited for automated reasoning methods and other applications (by exploiting general results from the labelled paradigm~\cite{CiaMafSpe13,DycNeg12,Neg16}). Such results have the added benefit that they expose interesting connections between the different proof-theoretic formalisms involved.








\ \\
\noindent
\acknowledgments{The author would like to express his gratitude to his supervisor A. Ciabattoni for her support and helpful comments.}

\bibliographystyle{abbrv}
\bibliography{bib}

\appendix

\section{Proofs}
\label{app:proofs}

In what follows, all results are proven with the assumption that $A[a/x] = A(a)$. Nevertheless, all results continue to hold in the vacuous setting as well (where $a$ does not occur in $A[a/x]$).

\begin{customthm}{\ref{thm:lint-properties}} The calculi $\lint$ and $\lintfocd$ have the following properties:
\begin{itemize}

\item[$(i)$] 

\begin{itemize}

\item[(a)] For all $A \in \mathcal{L}$, $ \vdash_{\lint} \R,w \leq v, w : A, \Gamma \Rightarrow v : A, \Delta$;

\item[(b)] For all $A \in \mathcal{L}$, $ \vdash_{\lint} \R,w:A,\Gamma \Rightarrow \Delta, w :A$; 

\item[(c)] For all $A \in \mathcal{L_{Q}}$, $\vdash_{\lintfocd} \R,w \leq v, \vv{a} \in D_{w}, w : A(\vv{a}), \Gamma \Rightarrow v : A(\vv{a}), \Delta$; 

\item[(d)] For all $A \in \mathcal{L_{Q}}$, $\vdash_{\lintfocd} \R, \vv{a} \in D_{w}, w:A(\vv{a}),\Gamma \Rightarrow \Delta, w :A(\vv{a})$;

\end{itemize}

\item[$(ii)$] The $(lsub)$ and $(psub)$ rules are height-preserving $($i.e. `hp-'$)$ admissible;
\begin{center}
\begin{tabular}{c @{\hskip 1em} c}
\AxiomC{$\R,\Gamma \Rightarrow \Delta$}
\RightLabel{$(lsub)$}
\UnaryInfC{$\R[w/v],\Gamma[w/v] \Rightarrow \Delta[w/v]$}
\DisplayProof

&

\AxiomC{$\R,\Gamma \Rightarrow \Delta$}
\RightLabel{$(psub)$}
\UnaryInfC{$\R[a/b],\Gamma[a/b] \Rightarrow \Delta[a/b]$}
\DisplayProof
\end{tabular}
\end{center}


\item[$(iii)$] All rules are hp-invertible;

\item[$(iv)$] The $\wk$ and $\{(ctr_{R}),(ctr_{F_{l}}),(ctr_{F_{r}})\}$ rules (below) are hp-admissible;

\begin{center}
\begin{tabular}{c @{\hskip 2em} c} 
\AxiomC{$\R,\Gamma \Rightarrow \Delta$}
\RightLabel{$(wk)$}
\UnaryInfC{$\R,\R',\Gamma',\Gamma \Rightarrow \Delta',\Delta$}
\DisplayProof

&

\AxiomC{$\R,\R',\R',\Gamma \Rightarrow \Delta$}
\RightLabel{$(ctr_{R})$}
\UnaryInfC{$\R,\R',\Gamma \Rightarrow \Delta$}
\DisplayProof

\end{tabular}
\end{center}
\begin{center}
\begin{tabular}{c @{\hskip 2em} c} 

\AxiomC{$\R,\Gamma',\Gamma',\Gamma \Rightarrow \Delta$}
\RightLabel{$(ctr_{F_{l}})$}
\UnaryInfC{$\R,\Gamma',\Gamma \Rightarrow \Delta$}
\DisplayProof

&

\AxiomC{$\R,\Gamma \Rightarrow \Delta, \Delta', \Delta'$}
\RightLabel{$(ctr_{F_{r}})$}
\UnaryInfC{$\R,\Gamma \Rightarrow \Delta, \Delta'$}
\DisplayProof
\end{tabular}
\end{center}

\item[$(v)$] The $\cut$ rule (below) is admissible;

\begin{center}
\begin{tabular}{c}
\AxiomC{$\R,\Gamma \Rightarrow \Delta, w :A$}
\AxiomC{$\R,w :A,\Gamma \Rightarrow \Delta$}
\RightLabel{$\cut$}
\BinaryInfC{$\R,\Gamma \Rightarrow \Delta$}
\DisplayProof
\end{tabular}
\end{center}

\item[$(vi)$] $\lint$ $(\lintfocd)$ is sound and complete for $\int$ $(\intfocd$, resp.$)$.

\end{itemize}
\end{customthm}

\begin{proof} We argue that properties (i)--(vi) obtain for $\lintfocd$.


\textit{Claim (i).} Claims (a) and (b) are shown in~\cite[Lem.~5.1]{DycNeg12}, so we focus on proving (c) and (d). We prove claims (c) and (d) together by mutual induction on the complexity of $A$.\\

\textit{Claim (i)-(c).} The base case is resolved using the $\idfo$ rule for atomic formulae and the $\botr$ rule for $\bot$. We provide the cases for $\imp$, $\exists$, and $\forall$ for the inductive step; the $\lor$ and $\land$ cases are simple to verify. Note that for the $\imp$ case the parameters $\vv{a_{1}}$ and  $\vv{a_{2}}$ are all and only those parameters that occur in $\vv{a}$, with $\vv{a_{1}}$ and  $\vv{a_{2}}$ potentially intersecting. Also, the $\imp$ case relies on claim (d) (shown below).

\begin{center}
\resizebox{\columnwidth}{!}{
\begin{tabular}{c}
\AxiomC{$\Pi_{1}$}
\AxiomC{$\Pi_{2}$}
\RightLabel{$\imprl$}
\BinaryInfC{$\R, w \leq v, v \leq u, w \leq u, \vv{a} \in D_{w}, w : A(\vv{a_{1}}) \imp B(\vv{a_{2}}), u : A(\vv{a_{1}}), \Gamma \Rightarrow u : B(\vv{a_{2}}), \Delta$}
\RightLabel{$\trans$}
\UnaryInfC{$\R, w \leq v, v \leq u, \vv{a} \in D_{w}, w : A(\vv{a_{1}}) \imp B(\vv{a_{2}}), u : A(\vv{a_{1}}), \Gamma \Rightarrow u : B(\vv{a_{2}}), \Delta$}
\RightLabel{$\imprr$}
\UnaryInfC{$\R, w \leq v, \vv{a} \in D_{w}, w : A(\vv{a_{1}}) \imp B(\vv{a_{2}}), \Gamma \Rightarrow v : A(\vv{a_{1}}) \imp B(\vv{a_{2}}), \Delta$}
\DisplayProof
\end{tabular}
}
\end{center}

\begin{center}
\resizebox{\columnwidth}{!}{
\begin{tabular}{c c c}
$\Pi_{1}$

&

$= \Big \{$

&

\AxiomC{}
\RightLabel{(d)}
\dashedLine
\UnaryInfC{$\R, w \leq v, v \leq u, w \leq u, \vv{a} \in D_{w}, w : A(\vv{a_{1}}) \imp B(\vv{a_{2}}), u : A(\vv{a_{1}}), \Gamma \Rightarrow u : A(\vv{a_{1}}), u : B(\vv{a_{2}}), \Delta$}
\DisplayProof
\end{tabular}
}
\end{center}
\begin{center}
\resizebox{\columnwidth}{!}{
\begin{tabular}{c c c}
$\Pi_{2}$

&

$= \Big \{$

&

\AxiomC{}
\RightLabel{(d)}
\dashedLine
\UnaryInfC{$\R, w \leq v, v \leq u, w \leq u, \vv{a} \in D_{w}, w : A(\vv{a_{1}}) \imp B(\vv{a_{2}}), u : A(\vv{a_{1}}),  u : B(\vv{a_{2}}), \Gamma \Rightarrow u : B(\vv{a_{2}}), \Delta$}
\DisplayProof
\end{tabular}
}
\end{center}
\begin{center}
\begin{tabular}{c} 
\AxiomC{}
\RightLabel{IH}
\dashedLine
\UnaryInfC{$\R, w \leq v, b \in D_{w}, \vv{a} \in D_{w}, w : A(b,\vv{a}), \Gamma \Rightarrow v : A(b,\vv{a}), \Delta$}
\RightLabel{$\existsr$}
\UnaryInfC{$\R, w \leq v, b \in D_{w}, \vv{a} \in D_{w}, w : A(b,\vv{a}), \Gamma \Rightarrow v : \exists x A(x,\vv{a}), \Delta$}
\RightLabel{$\existsl$}
\UnaryInfC{$\R, w \leq v, \vv{a} \in D_{w}, w : \exists x A(x,\vv{a}), \Gamma \Rightarrow v : \exists x A(x,\vv{a}), \Delta$}
\DisplayProof
\end{tabular}
\end{center}
\begin{center}
\resizebox{\columnwidth}{!}{
\begin{tabular}{c}
\AxiomC{}
\RightLabel{IH}
\dashedLine
\UnaryInfC{$\R, w \leq v, v \leq u, b \in D_{u}, \vv{a} \in D_{w}, v : A(b,\vv{a}), w : \forall x A(x,\vv{a}), \Gamma \Rightarrow u : A(b,\vv{a}), \Delta$}
\RightLabel{$\alll$}
\UnaryInfC{$\R, w \leq v, v \leq u, b \in D_{u}, \vv{a} \in D_{w}, w : \forall x A(x,\vv{a}), \Gamma \Rightarrow u : A(b,\vv{a}), \Delta$}
\RightLabel{$\allr$}
\UnaryInfC{$\R, w \leq v, \vv{a} \in D_{w}, w : \forall x A(x,\vv{a}), \Gamma \Rightarrow v : \forall x A(x,\vv{a}), \Delta$}
\DisplayProof
\end{tabular}
}
\end{center}

\textit{Claim (i)-(d).} The base case for atomic formulae is shown below; the case for $\bot$ is omitted as it is simple to verify using the $\botr$ rule. We provide the $\imp$, $\exists$, and $\forall$ cases as the cases for the other connectives are straightforward to prove.\\ 

\begin{center}
\begin{tabular}{c}
\AxiomC{}
\RightLabel{$\idfo$}
\UnaryInfC{$\R, w \leq w, \vv{a} \in D_{w}, w : p(\vv{a}), \Gamma \Rightarrow w : p(\vv{w}), \Delta$}
\RightLabel{$\refl$}
\UnaryInfC{$\R, \vv{a} \in D_{w}, w : p(\vv{a}), \Gamma \Rightarrow w : p(\vv{a}), \Delta$}
\DisplayProof
\end{tabular}
\end{center}
\begin{center}
\begin{tabular}{c} 
\AxiomC{$\Pi_{1}$}
\AxiomC{$\Pi_{2}$}
\RightLabel{$\imprl$}
\BinaryInfC{$\R, w \leq v, \vv{a} \in D_{w}, w : A(\vv{a_{1}}) \imp B(\vv{a_{2}}), v : A(\vv{a_{1}}), \Gamma \Rightarrow v : B(\vv{a_{2}}), \Delta$}
\RightLabel{$\imprr$}
\UnaryInfC{$\R, \vv{a} \in D_{w}, w : A(\vv{a_{1}}) \imp B(\vv{a_{2}}), \Gamma \Rightarrow w : A(\vv{a_{1}}) \imp B(\vv{a_{2}}), \Delta$}
\DisplayProof
\end{tabular}
\end{center}
\begin{center}
\resizebox{\columnwidth}{!}{
\begin{tabular}{c c c}
$\Pi_{1}$

&

$= \Big \{$

&

\AxiomC{}
\RightLabel{IH}
\dashedLine
\UnaryInfC{$\R, w \leq v, \vv{a} \in D_{w}, w : A(\vv{a_{1}}) \imp B(\vv{a_{2}}), v : A(\vv{a_{1}}), \Gamma \Rightarrow v : A(\vv{a_{1}}), v : B(\vv{a_{2}}), \Delta$}
\DisplayProof
\end{tabular}
}
\end{center}
\begin{center}
\resizebox{\columnwidth}{!}{
\begin{tabular}{c c c}
$\Pi_{2}$

&

$= \Big \{$

&

\AxiomC{}
\RightLabel{IH}
\dashedLine
\UnaryInfC{$\R, w \leq v, \vv{a} \in D_{w}, w : A(\vv{a_{1}}) \imp B(\vv{a_{2}}), v : A(\vv{a_{1}}), v : B(\vv{a_{2}}), \Gamma \Rightarrow v : B(\vv{a_{2}}), \Delta$}
\DisplayProof
\end{tabular}
}
\end{center}
\begin{center}
\begin{tabular}{c} 
\AxiomC{}
\RightLabel{IH}
\dashedLine
\UnaryInfC{$\R, b \in D_{w}, \vv{a} \in D_{w}, w : A(b,\vv{a}), \Gamma \Rightarrow w : A(b,\vv{a}), \Delta$}
\RightLabel{$\existsr$}
\UnaryInfC{$\R, b \in D_{w}, \vv{a} \in D_{w}, w : A(b,\vv{a}), \Gamma \Rightarrow w : \exists x A(x,\vv{a}), \Delta$}
\RightLabel{$\existsl$}
\UnaryInfC{$\R, \vv{a} \in D_{w}, w : \exists x A(x,\vv{a}), \Gamma \Rightarrow w : \exists x A(x,\vv{a}), \Delta$}
\DisplayProof
\end{tabular}
\end{center}

\begin{center}
\resizebox{\columnwidth}{!}{
\begin{tabular}{c}
\AxiomC{}
\RightLabel{IH}
\dashedLine
\UnaryInfC{$\R, w \leq v, b \in D_{u}, \vv{a} \in D_{w}, v : A(b,\vv{a}), w : \forall x A(x,\vv{a}), \Gamma \Rightarrow v : A(b,\vv{a}), \Delta$}
\RightLabel{$\alll$}
\UnaryInfC{$\R, w \leq v, b \in D_{v}, \vv{a} \in D_{w}, w : \forall x A(x,\vv{a}), \Gamma \Rightarrow v : A(b,\vv{a}), \Delta$}
\RightLabel{$\allr$}
\UnaryInfC{$\R, \vv{a} \in D_{w}, w : \forall x A(x,\vv{a}), \Gamma \Rightarrow w : \forall x A(x,\vv{a}), \Delta$}
\DisplayProof
\end{tabular}
}
\end{center}

\textit{Claim (ii).} Hp-admissibility of $(lsub)$ and $(psub)$ are proven by induction on the height of the given derivation; both are similar to \cite[Lem.~5.1]{DycNeg12}. The only non-trivial cases for the former are the $\imprr$ and $\allr$ rules of the inductive step, and the non-trivial cases for the latter are the $\allr$ and $\existsl$ rules of the inductive step. In such cases the side condition must be preserved if the rule is to be applied. Nevertheless, this can be ensured in the usual way (cf. \cite[Lem.~5.1]{DycNeg12}) by evoking IH twice: first, IH replaces the eigenvariable of the $\imprr$, $\allr$, or $\existsl$ rule with a fresh parameter, and second, the substitution of $(lsub)$ or $(psub)$ is performed, after which, the corresponding rule may be applied to derive the desired conclusion.\\

\textit{Claim (iii).} The cases for the propositional rules are straightforward to check, so we omit them. By hp-admissibility of $\wk$ (property (v)), we know that $\existsr$, $\alll$, $\nd$, and $\cd$ are hp-invertible (note that the proof of $\wk$ admissibility does not depend on property (iii) holding, so we may evoke it). We therefore only need to check that $\allr$ and $ \existsl$ are hp-invertible. We prove the claim by induction on the height of the given derivation of $\R, \Gamma \Rightarrow \Delta, w : \forall x A(x)$ for the $\allr$ rule; the proof for $\existsl$ is similar.\\

\textit{Base case.} If the height of the derivation is $0$, then $\R, w \leq v, a \in D_{v}, \Gamma \Rightarrow \Delta, v : A(a)$ is either an instance of $\botr$ or $\idfo$.\\

\textit{Inductive step.} If the last rule applied in the derivation is $\conrl$, $\conrr$, $\disrl$, $\disrr$, $\imprl$, $\refl$, $\trans$, $\existsr$, $\alll$, $\nd$, or $\cd$, then the conclusion follows by applying IH followed by an application of the associated rule. If the last rule applied is $\imprr$, $\existsl$, or $\allr$ (where the principal formula of $\allr$ is in $\Delta$), then we potentially apply hp-admissibility of $(lsub)$ or $(psub)$ (property (ii) above), followed by IH, and then an application of the corresponding rule. If the last rule applied is $\allr$, with $w : \forall x A(x)$ the principal formula, then the premise of the inference gives the desired result.\\

\textit{Claim (iv).} We prove the hp-admissibility of $\wk$ and each contraction rule in the set $\{(ctr_{R}),(ctr_{F_{l}}),(ctr_{F_{r}})\}$ by induction on the height of the given derivation. Hp-admissibility of $\wk$ is relatively straightforward; the only non-trivial cases are the $\imprr$, $\allr$, and $\existsl$ rules of the inductive step. Such cases are resolved, however, by potentially applying hp-admissibility of $(lsub)$ or $(psub)$ (property (ii) above), then IH, and last the corresponding rule. Hp-admissibility of $(ctr_{R})$ is simple; any application of the rule to an initial sequent yields an initial sequent, and every case in the inductive step is resolved by applying IH followed by the corresponding rule. The $(ctr_{F_{l}})$ and $(ctr_{F_{r}})$ cases are also quite straightforward. The only non-trivial cases occur when a principal occurrence of an $\imprr$, $\allr$, or $\existsl$ rule is contracted. In such cases, the inductive step is solved by evoking the hp-invertibility of each rule (property (iii) above), followed by an application of IH, and hp-admissibility of $(ctr_{R})$ if needed.\\

\textit{Claim (v).} We prove the admissibility of $\cut$ by induction on the lexicographic ordering of tuples $(|A|,h_{1},h_{2})$, where $|A|$ is the complexity of the cut formula $A$, $h_{1}$ is the height of the left premise of $\cut$, and $h_{2}$ is the height of the right premise of $\cut$. We only consider the cases where the cut formula is prinicpal in both premises of $\cut$; the other cases where the cut formula is not principal in one premise, or in both premises, are relatively straightforward to verify, though the large number of cases makes the proof tedious. By the cut admissibility theorem for $\lint$~\cite[Thm.~5.6]{DycNeg12}, we need only verify the cases where the cut formula is of the form $\forall x B(x)$ or $\exists x B(x)$. We show each case in turn below:

\begin{center}
\resizebox{\columnwidth}{!}{
\begin{tabular}{c} 
\AxiomC{$\R, w \leq v, a \in D_{v}, w \leq u, b \in D_{u}, \Gamma \Rightarrow \Delta, u : B(b)$}
\RightLabel{$\allr$}
\UnaryInfC{$\R, w \leq v, a \in D_{v}, \Gamma \Rightarrow \Delta, w : \forall x B(x)$}

\AxiomC{$\R, w \leq v, a \in D_{v}, v : B(a), w : \forall x B(x), \Gamma \Rightarrow \Delta$}
\RightLabel{$\alll$}
\UnaryInfC{$\R, w \leq v, a \in D_{v}, w : \forall x B(x), \Gamma \Rightarrow \Delta$}

\RightLabel{$\cut$}
\BinaryInfC{$\R, w \leq v, a \in D_{v}, \Gamma \Rightarrow \Delta$}

\DisplayProof
\end{tabular}
}
\end{center}

The case is resolved as follows:

\begin{center}
\begin{tabular}{c}
\AxiomC{$\Pi_{1}$}
\AxiomC{$\Pi_{2}$}
\RightLabel{$\cut$}
\BinaryInfC{$\R, w \leq v, a \in D_{v}, \Gamma \Rightarrow \Delta$}
\DisplayProof
\end{tabular}
\end{center}

\begin{center}
\begin{tabular}{c c c}
$\Pi_{1}$

&

$= \Bigg \{$

&

\AxiomC{$\R, w \leq u, b \in D_{u}, \Gamma \Rightarrow \Delta, u : B(b)$}
\RightLabel{$(lsub)$}
\dashedLine
\UnaryInfC{$\R, w \leq v, b \in D_{v}, \Gamma \Rightarrow \Delta, v : B(b)$}
\RightLabel{$(psub)$}
\dashedLine
\UnaryInfC{$\R, w \leq v, a \in D_{v}, \Gamma \Rightarrow \Delta, v : B(a)$}
\DisplayProof
\end{tabular}
\end{center}

\begin{center}
\resizebox{\columnwidth}{!}{
\begin{tabular}{c c c}
$\Pi_{2}$

&

$= \Bigg \{$

&

\AxiomC{$\R, w \leq u, b \in D_{u}, \Gamma \Rightarrow \Delta, u : B(b)$}
\RightLabel{$\wk$}
\dashedLine
\UnaryInfC{$\R, w \leq v, a \in D_{v}, v : B(a), w \leq u, b \in D_{u}, \Gamma \Rightarrow \Delta, u : B(b)$}
\RightLabel{$\allr$}
\UnaryInfC{$\R, w \leq v, a \in D_{v}, v : B(a), \Gamma \Rightarrow \Delta, w : \forall x B(x)$}

\AxiomC{$\R, w \leq v, a \in D_{v}, v : B(a), w : \forall x B(x), \Gamma \Rightarrow \Delta$}

\RightLabel{$\cut$}
\BinaryInfC{$\R, w \leq v, a \in D_{v}, v : B(a), \Gamma \Rightarrow \Delta$}
\DisplayProof
\end{tabular}
}
\end{center}
Observe that the $\cut$ in $\Pi_{2}$ has a height $h_{2}$ that is one less than the original $\cut$, and the second $\cut$ is on a formula $B(a)$ that is of less complexity than the original $\cut$. Let us examine the $\exists x B(x)$ case.

\begin{center}
\resizebox{\columnwidth}{!}{
\begin{tabular}{c}
\AxiomC{$\R, a \in D_{w}, \Gamma \Rightarrow \Delta, w: B(a), w: \exists x B(x)$}
\RightLabel{$\existsr$}
\UnaryInfC{$\R, a \in D_{w}, \Gamma \Rightarrow \Delta, w: \exists x B(x)$}

\AxiomC{$\R, a \in D_{w}, b \in D_{w}, w: B(b), \Gamma \Rightarrow \Delta$}
\RightLabel{$\existsl$}
\UnaryInfC{$\R, a \in D_{w}, w : \exists x B(x), \Gamma \Rightarrow \Delta$}

\RightLabel{$\cut$}

\BinaryInfC{$\R, a \in D_{w}, \Gamma \Rightarrow \Delta$}
\DisplayProof
\end{tabular}
}
\end{center}
The case is resolved as follows:
\begin{center}
\begin{tabular}{c}
\AxiomC{$\Pi_{1}$}
\AxiomC{$\Pi_{2}$}
\RightLabel{$\cut$}
\BinaryInfC{$\R, a \in D_{w}, \Gamma \Rightarrow \Delta$}
\DisplayProof
\end{tabular}
\end{center}

\begin{center}
\resizebox{\columnwidth}{!}{
\begin{tabular}{c c c}
$\Pi_{1}$

&

$= \Bigg \{ \qquad$

&

\AxiomC{$\R, a \in D_{w}, \Gamma \Rightarrow \Delta, w: B(a), w: \exists x B(x)$}

\AxiomC{$\R, a \in D_{w}, b \in D_{w}, w: B(b), \Gamma \Rightarrow \Delta$}
\RightLabel{$\wk$}
\dashedLine
\UnaryInfC{$\R, a \in D_{w}, b \in D_{w}, w: B(b), \Gamma \Rightarrow w : B(a), \Delta$}
\RightLabel{$\existsl$}
\UnaryInfC{$\R, a \in D_{w}, w : \exists x B(x), \Gamma \Rightarrow w : B(a), \Delta$}

\RightLabel{$\cut$}

\BinaryInfC{$\R, a \in D_{w}, \Gamma \Rightarrow w : B(a), \Delta$}

\DisplayProof
\end{tabular}
}
\end{center}

\begin{center}
\begin{tabular}{c c c}
$\Pi_{2}$

&

$= \Bigg \{$

&

\AxiomC{$\R, a \in D_{w}, b \in D_{w}, w: B(b), \Gamma \Rightarrow \Delta$}
\RightLabel{$(psub)$}
\dashedLine
\UnaryInfC{$\R, a \in D_{w}, a \in D_{w}, w: B(a), \Gamma \Rightarrow \Delta$}
\RightLabel{$(ctr_{R})$}
\dashedLine
\UnaryInfC{$\R, a \in D_{w}, w: B(a), \Gamma \Rightarrow \Delta$}
\DisplayProof
\end{tabular}
\end{center}
Observe that the $\cut$ in the $\Pi_{1}$ inference has a height $h_{1}$ that is one less than the original $\cut$, and the second $\cut$ is on a formula of smaller complexity.\\

\textit{Claim (vi).} Soundness is shown by interpreting labelled sequents on Kripke models for $\intfocd$~\cite[Ch.~3]{GabSheSkv09} and proving that validity is preserved from premise to conclusion. By~\cite{DycNeg12}, we know that $\lint$ is complete relative to $\int$, so we need only show that $\lintfocd$ can derive quantifier axioms and simulate the inference rules of the axiomatization for $\intfocd$ (see~\cite[Ch.~2.6]{GabSheSkv09}). We say that a formula $A(\vv{a})$ is derivable in $\lintfocd$ if and only if $\vv{a} \in D_{w} \Rightarrow w : A(\vv{a})$ is derivable in $\lintfocd$.


\begin{center}
\resizebox{\columnwidth}{!}{
\begin{tabular}{c} 
\AxiomC{}
\RightLabel{Prop.~(i)}
\dashedLine
\UnaryInfC{$w \leq u, u \leq u, \vv{a} \in D_{w}, a \in D_{w}, a \in D_{u}, \vv{a} \in D_{u}, u : \forall x A(\vv{a},x), u :A(\vv{a},a) \Rightarrow u : A(\vv{a},a)$}
\RightLabel{$\alll$}
\UnaryInfC{$w \leq u, u \leq u, \vv{a} \in D_{w}, a \in D_{w}, a \in D_{u}, \vv{a} \in D_{u}, u : \forall x A(\vv{a},x) \Rightarrow u : A(\vv{a},a)$}
\RightLabel{$\nd \times k$}
\UnaryInfC{$w \leq u, u \leq u, \vv{a} \in D_{w}, a \in D_{w}, u : \forall x A(\vv{a},x) \Rightarrow u : A(\vv{a},a)$}
\RightLabel{$\refl$}
\UnaryInfC{$w \leq u, \vv{a} \in D_{w}, a \in D_{w}, u : \forall x A(\vv{a},x) \Rightarrow u : A(\vv{a},a)$}
\RightLabel{$\imprr$}
\UnaryInfC{$\vv{a} \in D_{w},a \in D_{w} \Rightarrow w : \forall x A(\vv{a},x) \imp A(\vv{a},a)$}
\DisplayProof
\end{tabular}
}
\end{center}

\begin{center}
\resizebox{\columnwidth}{!}{
\begin{tabular}{c}
\AxiomC{}
\RightLabel{Prop.~(i)}
\dashedLine
\UnaryInfC{$w \leq u, \vv{a} \in D_{w},a \in D_{w}, \vv{a} \in D_{u}, a \in D_{u}, u : A(a) \Rightarrow u : A(\vv{a},a), u : \exists x A(\vv{a},x)$}
\RightLabel{$\existsr$}
\UnaryInfC{$w \leq u, \vv{a} \in D_{w},a \in D_{w}, \vv{a} \in D_{u}, a \in D_{u}, u : A(\vv{a},a) \Rightarrow u : \exists x A(\vv{a},x)$}
\RightLabel{$\nd \times k$}
\UnaryInfC{$w \leq u, \vv{a} \in D_{w},a \in D_{w},  u : A(\vv{a},a) \Rightarrow u : \exists x A(\vv{a},x)$}
\RightLabel{$\imprr$}
\UnaryInfC{$\vv{a} \in D_{w},a \in D_{w} \Rightarrow w : A(\vv{a},a) \imp \exists x A(\vv{a},x)$}
\DisplayProof
\end{tabular}
}
\end{center}

\begin{center}
\resizebox{\columnwidth}{!}{
\begin{tabular}{c}
\AxiomC{$\Pi_{1}$}
\AxiomC{$\Pi_{2}$}
\RightLabel{$\imprl$}
\BinaryInfC{$w \leq v, v \leq u, u \leq z, v \leq z, \vv{a} \in D_{w}, \vv{b} \in D_{w}, a \in D_{z}, z : B(\vv{b}) \imp A(\vv{a},a), v : \forall x (B(\vv{b}) \imp A(\vv{a},x)), u : B(\vv{b}) \Rightarrow z : A(\vv{a},a)$}
\RightLabel{$\alll$}
\UnaryInfC{$w \leq v, v \leq u, u \leq z, v \leq z, \vv{a} \in D_{w}, \vv{b} \in D_{w}, a \in D_{z}, v : \forall x (B(\vv{b}) \imp A(\vv{a},x)), u : B(\vv{b}) \Rightarrow z : A(\vv{a},a)$}
\RightLabel{$\trans$}
\UnaryInfC{$w \leq v, v \leq u, u \leq z, \vv{a} \in D_{w}, \vv{b} \in D_{w}, a \in D_{z}, v : \forall x (B(\vv{b}) \imp A(\vv{a},x)), u : B(\vv{b}) \Rightarrow z : A(\vv{a},a)$}
\RightLabel{$\allr$}
\UnaryInfC{$w \leq v, v \leq u, \vv{a} \in D_{w}, \vv{b} \in D_{w}, v : \forall x (B(\vv{b}) \imp A(\vv{a},x)), u : B(\vv{b}) \Rightarrow u : \forall x A(\vv{a},x)$}
\RightLabel{$\imprr$}
\UnaryInfC{$w \leq v, \vv{a} \in D_{w}, \vv{b} \in D_{w}, v : \forall x (B(\vv{b}) \imp A(\vv{a},x)) \Rightarrow v : B(\vv{b}) \imp \forall x A(\vv{a},x)$}
\RightLabel{$\imprr$}
\UnaryInfC{$\vv{a} \in D_{w}, \vv{b} \in D_{w} \Rightarrow w : \forall x (B(\vv{b}) \imp A(\vv{a},x)) \imp (B(\vv{b}) \imp \forall x A(\vv{a},x))$}
\DisplayProof
\end{tabular}
}
\end{center}
The proofs $\Pi_{1}$ and $\Pi_{2}$ are as follows (resp.).

\begin{center}
\resizebox{\columnwidth}{!}{
\begin{tabular}{c}
\AxiomC{}
\RightLabel{Prop.~(i)}
\dashedLine
\UnaryInfC{$w \leq v, v \leq u, u \leq z, v \leq z, \vv{a} \in D_{w}, \vv{b} \in D_{w}, \vv{b} \in D_{u}, a \in D_{z}, z : B(\vv{b}) \imp A(\vv{a},a), v : \forall x (B(\vv{b}) \imp A(\vv{a},x)), u : B(\vv{b}) \Rightarrow z : B(\vv{b}), z : A(\vv{a},a)$}
\RightLabel{$\nd \times k_{1}$}
\UnaryInfC{$w \leq v, v \leq u, u \leq z, v \leq z, \vv{a} \in D_{w}, \vv{b} \in D_{w}, a \in D_{z}, z : B(\vv{b}) \imp A(\vv{a},a), v : \forall x (B(\vv{b}) \imp A(\vv{a},x)), u : B(\vv{b}) \Rightarrow z : B(\vv{b}), z : A(\vv{a},a)$}
\DisplayProof
\end{tabular}
}
\end{center}

\begin{center}
\resizebox{\columnwidth}{!}{
\begin{tabular}{c}
\AxiomC{}
\RightLabel{Prop.~(i)}
\dashedLine
\UnaryInfC{$w \leq v, v \leq u, u \leq z, v \leq z, a \in D_{z}, \vv{a} \in D_{w}, \vv{a} \in D_{z}, \vv{b} \in D_{w}, z : B(\vv{b}) \imp A(\vv{a},a), v : \forall x (B(\vv{b}) \imp A(\vv{a},x)), u : B(\vv{b}), z : A(\vv{a},a) \Rightarrow z : A(\vv{a},a)$}
\RightLabel{$\nd \times k_{2}$}
\UnaryInfC{$w \leq v, v \leq u, u \leq z, v \leq z, a \in D_{z}, \vv{a} \in D_{w}, \vv{b} \in D_{w}, z : B(\vv{b}) \imp A(\vv{a},a), v : \forall x (B(\vv{b}) \imp A(\vv{a},x)), u : B(\vv{b}), z : A(\vv{a},a) \Rightarrow z : A(\vv{a},a)$}
\DisplayProof
\end{tabular}
}
\end{center}

The proof of the axiom $\forall x (A(x) \imp B) \imp (\exists x A(x) \imp B)$ is similar to the previous proof. The generalization rule is simulated as shown below:\\

\begin{center}
\begin{tabular}{c}
\AxiomC{$\vv{a} \in D_{w}, a \in D_{w} \Rightarrow w : A(\vv{a},a)$}
\RightLabel{$\wk$}
\dashedLine
\UnaryInfC{$u \leq w, \vv{a} \in D_{u}, \vv{a} \in D_{w}, a \in D_{w} \Rightarrow w : A(\vv{a},a)$}
\RightLabel{$\nd \times k$}
\UnaryInfC{$u \leq w, \vv{a} \in D_{u}, a \in D_{w} \Rightarrow w : A(\vv{a},a)$}
\RightLabel{$\allr$}
\UnaryInfC{$\vv{a} \in D_{u}, \Rightarrow u : \forall x A(\vv{a},x)$}
\RightLabel{$(lsub)$}
\dashedLine
\UnaryInfC{$\vv{a} \in D_{w}, \Rightarrow w : \forall x A(\vv{a},x)$}
\DisplayProof
\end{tabular}
\end{center}

\begin{center}
\resizebox{\columnwidth}{!}{
\begin{tabular}{c}
\AxiomC{$\Pi_{1}$}
\AxiomC{$\Pi_{2}$}
\RightLabel{$\disrl$}
\BinaryInfC{$w \leq v, v \leq u, \vv{a} \in D_{w}, \vv{b} \in D_{w}, a \in D_{u}, v : A(\vv{a},a) \lor B(\vv{b}), v : \forall x (A(\vv{a},x) \lor B(\vv{b})) \Rightarrow u : A(\vv{a},a), v : B(\vv{b})$}
\RightLabel{$\alll$}
\UnaryInfC{$w \leq v, v \leq u, v \leq v, \vv{a} \in D_{w}, \vv{b} \in D_{w}, a \in D_{u}, v : \forall x (A(\vv{a},x) \lor B(\vv{b})) \Rightarrow u : A(\vv{a},a), v : B(\vv{b})$}
\RightLabel{$\refl$}
\UnaryInfC{$w \leq v, v \leq u, \vv{a} \in D_{w}, \vv{b} \in D_{w}, a \in D_{u}, v : \forall x (A(\vv{a},x) \lor B(\vv{b})) \Rightarrow u : A(\vv{a},a), v : B(\vv{b})$}
\RightLabel{$\allr$}
\UnaryInfC{$w \leq v, \vv{a} \in D_{w}, \vv{b} \in D_{w}, v : \forall x (A(\vv{a},x) \lor B(\vv{b})) \Rightarrow v : \forall x A(\vv{a},x), v : B(\vv{b})$}
\RightLabel{$\disrr$}
\UnaryInfC{$w \leq v, \vv{a} \in D_{w}, \vv{b} \in D_{w}, v : \forall x (A(\vv{a},x) \lor B(\vv{b})) \Rightarrow v : \forall x A(\vv{a},x) \lor B(\vv{b})$}
\RightLabel{$\imprr$}
\UnaryInfC{$\vv{a} \in D_{w}, \vv{b} \in D_{w} \Rightarrow w : \forall x (A(\vv{a},x) \lor B(\vv{b})) \imp \forall x A(\vv{a},x) \lor B(\vv{b})$}
\DisplayProof
\end{tabular}
}
\end{center}
The proofs $\Pi_{1}$ and $\Pi_{2}$ are as follows (resp.).

\begin{center}
\resizebox{\columnwidth}{!}{
\begin{tabular}{c}
\AxiomC{}
\RightLabel{Prop.~(i)}
\dashedLine
\UnaryInfC{$w \leq v, v \leq u, \vv{a} \in D_{w}, \vv{a} \in D_{v}, \vv{b} \in D_{w}, a \in D_{u}, a \in D_{v}, v : A(\vv{a},a), v : \forall x (A(\vv{a},x) \lor B(\vv{b})) \Rightarrow u : A(\vv{a},a), v : B(\vv{b})$}
\RightLabel{$\cd$}
\UnaryInfC{$w \leq v, v \leq u, \vv{a} \in D_{w}, \vv{a} \in D_{v}, \vv{b} \in D_{w}, a \in D_{u}, v : A(\vv{a},a), v : \forall x (A(\vv{a},x) \lor B(\vv{b})) \Rightarrow u : A(\vv{a},a), v : B(\vv{b})$}
\RightLabel{$\nd \times k_{1}$}
\UnaryInfC{$w \leq v, v \leq u, \vv{a} \in D_{w}, \vv{b} \in D_{w}, a \in D_{u}, v : A(\vv{a},a), v : \forall x (A(\vv{a},x) \lor B(\vv{b})) \Rightarrow u : A(\vv{a},a), v : B(\vv{b})$}
\DisplayProof
\end{tabular}
}
\end{center}

\begin{center}
\resizebox{\columnwidth}{!}{
\begin{tabular}{c}
\AxiomC{}
\RightLabel{Prop.~(i)}
\dashedLine
\UnaryInfC{$w \leq v, v \leq u, \vv{a} \in D_{w}, \vv{b} \in D_{w}, \vv{b} \in D_{v}, a \in D_{u}, v : B(\vv{b}), v : \forall x (A(\vv{a},x) \lor B(\vv{b})) \Rightarrow u : A(\vv{a},a), v : B(\vv{b})$}
\RightLabel{$\nd \times k_{2}$}
\UnaryInfC{$w \leq v, v \leq u, \vv{a} \in D_{w}, \vv{b} \in D_{w}, a \in D_{u}, v : B(\vv{b}), v : \forall x (A(\vv{a},x) \lor B(\vv{b})) \Rightarrow u : A(\vv{a},a), v : B(\vv{b})$}
\DisplayProof
\end{tabular}
}
\end{center}
\qed
\end{proof}


\begin{customlem}{\ref{lem:extended-lint-properties}} The calculus $\lint + \{(id^{*}), (\neg_{l}), (\neg_{r}), (\imp^{*}_{l}), \lift\}$ and the calculus $\lintfocd + \{\idfonew, (\neg_{l}), (\neg_{r}), (\imp^{*}_{l}), (\forall_{l}^{*}), (\forall^{*}_{r}), (\exists_{r}^{*}), \lift\}$ have the following properties:
\begin{itemize}

\item[$(i)$] All sequents of the form $\R,w \leq v, \vv{a} \in D_{w}, w : A(\vv{a}), \Gamma \Rightarrow v : A(\vv{a}), \Delta$ and $\R, \vv{a} \in D_{w}, w:A(\vv{a}),\Gamma \Rightarrow \Delta, w :A(\vv{a})$ are derivable;







\item[(ii)] The rules $\{(lsub),(psub),\wk,(ctr_{R}),(ctr_{F_{r}})\}$ are hp-admissible;

\item[(iii)] With the exception of $\{\conrl,\existsl\}$, all rules are hp-invertible;

\item[(iv)] The rules $\{\conrl,\existsl\}$ are invertible;

\item[(v)] The rule $(ctr_{F_{l}})$ is admissible.

\end{itemize}
\end{customlem}

\begin{proof} We prove that the properties hold for the first-order calculus since our method of proof implies that the properties hold in the propositional calculus.

\textit{Claim (i).} Same as Thm.~\ref{thm:lint-properties}-(i).

\textit{Claim (ii).} By induction on the height of the given derivation; similar to Thm.~\ref{thm:lint-properties}-(ii) and Thm.~\ref{thm:lint-properties}-(iv). Hp-admissibility of $(ctr_{F_{r}})$ evokes property (iii) below.

\textit{Claim (iii).} Hp-invertibility of $\imprl$, $\negl$, $\existsr$, $\alll$, $\refl$, $\trans$, $\nd$, $\cd$, $(\imp_{l}^{*})$, $(\forall_{l}^{*})$, $(\exists_{r}^{*})$, and $\lift$ follow from hp-admissibility of $\wk$ (property (ii) above). Hence, we need only prove invertibility of $\conrr$, $\disrl$, $\disrr$, $\imprr$, $\allr$, $\negr$, and $(\forall^{*}_{r})$ rules. The result is shown by induction on the height of the given derivation. The base cases are simple and every case of the inductive step is resolved by evoking IH followed by an application of the relevant rule.

\textit{Claim (iv).} The addition of $\lift$ and $\nd$ to our calculus breaks the \emph{height preserving} invertibility of $\conrl$ and $\existsl$. Nevertheless, it is worthwhile to note that the rule $(lift')$ ought to allow for hp-invertibility of $\conrl$ in the propositional calculus, and $(lnd)$ ought to allow for the hp-invertibility of $\conrl$ and $\nd$ in the first-order calculus (which would have the effect that $(ctr_{F_{l}})$ is hp-admissible in both calculi) while retaining the soundness and cut-free completeness of each system. The notation $v : \Gamma'$ is used to represent multisets of formulae labelled with $v$.
\begin{center}
\resizebox{\columnwidth}{!}{
\begin{tabular}{c c}
\AxiomC{$\R, w \leq u, \Gamma, w : \Gamma', u : \Gamma' \Rightarrow \Delta$}
\RightLabel{$(lift')$}
\UnaryInfC{$\R, w \leq u, \Gamma, w : \Gamma' \Rightarrow \Delta$}
\DisplayProof

&

\AxiomC{$\R, w \leq u, \vv{a} \in D_{w}, \vv{a} \in D_{u}, \Gamma, w : \Gamma', u : \Gamma' \Rightarrow \Delta$}
\RightLabel{$(lnd)$}
\UnaryInfC{$\R, w \leq u, \vv{a} \in D_{w}, \Gamma, w : \Gamma' \Rightarrow \Delta$}
\DisplayProof
\end{tabular}
}
\end{center}
Despite this shortcoming, the rules are still invertible. To show this, we prove the following two claims by induction on the height of the given derivation.\\

\begin{itemize}

\item[(a)] If $\R, w_{0} : A \land B, \ldots, w_{n} : A \land B, \Gamma \Rightarrow \Delta$ is provable, then so is the sequent $\R, w_{0} : A, w_{0} : B, \ldots, w_{n} : A, w_{n} : B, \Gamma \Rightarrow \Delta$.

\item[(b)] If $\R, w_{0} : \exists x A(x), \ldots, w_{n} : \exists x A(x), \Gamma \Rightarrow \Delta$ is provable, then so is the sequent $\R, a_{0} \in D_{w_{0}}, \ldots, a_{n} \in D_{w_{n}}, w_{0} : A(a_{0}), \ldots, w_{n} : A(a_{n}), \Gamma \Rightarrow \Delta$.

\end{itemize}

\textit{Claim (a).} The base case is trivial, so we move on to the inductive step.\\

\textit{Inductive step.} In all cases, with the exception of $\lift$, apply IH followed by the relevant rule. If none of the conjunctions are active in a $\lift$ inference, then apply IH followed by an application of $\lift$. If one of the conjunctions is principal in an application of $\lift$ (as shown below top), then resolve the case as shown below bottom:
\begin{center}
\AxiomC{$\R, u \leq w_{0}, u : A \land B, w_{0} : A \land B, \ldots, w_{n} : A \land B, \Gamma \Rightarrow \Delta$}
\RightLabel{$\lift$}
\UnaryInfC{$\R, u \leq w_{0}, u : A \land B, \ldots, w_{n} : A \land B, \Gamma \Rightarrow \Delta$}
\DisplayProof
\end{center}
\begin{center}
\AxiomC{}
\RightLabel{IH}
\dashedLine
\UnaryInfC{$\R, u \leq w_{0}, u : A, u :B, w_{0} : A, w_{0} : B, \ldots, w_{n} : A, w_{n} : B, \Gamma \Rightarrow \Delta$}
\RightLabel{$\lift$}
\UnaryInfC{$\R, u \leq w_{0}, u :A, u :B, w_{0} : B, \ldots, w_{n} : A, w_{n} : B, \Gamma \Rightarrow \Delta$}
\RightLabel{$\lift$}
\UnaryInfC{$\R, u \leq w_{0}, u : A, u :B, \ldots, w_{n} : A, w_{n} : B, \Gamma \Rightarrow \Delta$}
\DisplayProof
\end{center}
Notice that the two applications of $\lift$ needed to derive the desired conclusion break the hp-invertibility of the $\conrl$ rule.\\

\textit{Claim (b).} The base case is trivial, so we move on to the inductive step.\\

\textit{Inductive step.} All cases, with the exception of the one given below top (where one of our existential formulae is principal in an application of $\lift$), are resolved by applying IH followed by the corresponding rule. The non-trivial case given below top is resolved as shown below bottom.
\begin{center}
\AxiomC{$\R, u \leq w_{0}, u : \exists x A(x), w_{0} : \exists x A(x), \ldots, w_{n} : \exists x A(x), \Gamma \Rightarrow \Delta$}
\RightLabel{$\lift$}
\UnaryInfC{$\R, u \leq w_{0}, u : \exists x A(x), \ldots, w_{n} : \exists x A(x), \Gamma \Rightarrow \Delta$}
\DisplayProof
\end{center}
\begin{center}
\resizebox{\columnwidth}{!}{
\AxiomC{}
\RightLabel{IH}
\dashedLine
\UnaryInfC{$\R, u \leq w_{0}, a \in D_{u}, a_{0} \in D_{w_{0}}, \ldots, a_{n} \in D_{w_{n}}, u : A(a), w_{0} : A(a_{0}), \ldots, w_{n} : A(a_{n}), \Gamma \Rightarrow \Delta$}
\RightLabel{$(psub)$}
\dashedLine
\UnaryInfC{$\R, u \leq w_{0}, a \in D_{u}, a \in D_{w_{0}}, \ldots, a_{n} \in D_{w_{n}}, u : A(a), w_{0} : A(a), \ldots, w_{n} : A(a_{n}), \Gamma \Rightarrow \Delta$}
\RightLabel{$\lift$}
\UnaryInfC{$\R, u \leq w_{0}, a \in D_{u}, a \in D_{w_{0}}, \ldots, a_{n} \in D_{w_{n}}, u : A(a),\ldots, w_{n} : A(a_{n}), \Gamma \Rightarrow \Delta$}
\RightLabel{$\nd$}
\UnaryInfC{$\R, u \leq w_{0}, a \in D_{u}, \ldots, a_{n} \in D_{w_{n}}, u : A(a),\ldots, w_{n} : A(a_{n}), \Gamma \Rightarrow \Delta$}
\DisplayProof
}
\end{center}
Observe that the use of $\lift$ and $\nd$ breaks height-preserving invertibility of the rule.\\

Propositions (a) and (b) imply the invertibility of $\conrl$ and $\existsl$.\\

\textit{Claim (v).} By induction on pairs of the form $(|A|,h)$ where $|A|$ is the complexity of the contraction formula $A$ and $h$ is the height of the derivation. The proof makes use of properties (iii) and (iv).
\qed
\end{proof}

\begin{customlem}{\ref{lem:trans-refl-fo-elim}}
The $\refl$ and $\trans$ rules are admissible in the calculus $\intfocdl$.
\end{customlem}

\begin{proof} The $\trans$ elimination cases are given below.

\begin{center}
\begin{tabular}{c}
\AxiomC{}
\RightLabel{$\idfo$}
\UnaryInfC{$\R, w \leq u, u \leq v, w \leq v, \vv{a} \in D_{w}, w : p(\vv{a}), \Gamma \Rightarrow v : p(\vv{a}), \Delta$}
\RightLabel{$\trans$}
\UnaryInfC{$\R, w \leq u, u \leq v, \vv{a} \in D_{w}, w : p(\vv{a}), \Gamma \Rightarrow v : p(\vv{a}), \Delta$}
\DisplayProof
\end{tabular}
\end{center}

\begin{center}
\begin{tabular}{c}
\resizebox{\columnwidth}{!}{
\AxiomC{}
\RightLabel{$\idfo$}
\UnaryInfC{$\R, w \leq u, u \leq v, \vv{a} \in D_{w}, \vv{a} \in D_{u}, w : p(\vv{a}), u : p(\vv{a}), \Gamma \Rightarrow v : p(\vv{a}), \Delta$}
\RightLabel{$\lift + \nd \times n$}
\UnaryInfC{$\R, w \leq u, u \leq v, \vv{a} \in D_{w}, w : p(\vv{a}), \Gamma \Rightarrow v : p(\vv{a}), \Delta$}
\DisplayProof
}
\end{tabular}
\end{center}

\begin{center}
\AxiomC{$\R, w \leq u, u \leq v, w \leq v, a \in D_{v}, v : A[a/x], w : \forall x A, \Gamma \Rightarrow \Delta$}
\RightLabel{$\alll$}
\UnaryInfC{$\R, w \leq u, u \leq v, w \leq v, a \in D_{v}, w : \forall x A, \Gamma \Rightarrow \Delta$}
\RightLabel{$\trans$}
\UnaryInfC{$\R, w \leq u, u \leq v, a \in D_{v}, w : \forall x A, \Gamma \Rightarrow \Delta$}
\DisplayProof
\end{center}
\begin{center}
\AxiomC{}
\RightLabel{IH}
\dashedLine
\UnaryInfC{$\R, w \leq u, u \leq v, a \in D_{v}, v : A[a/x], w : \forall x A, \Gamma \Rightarrow \Delta$}
\RightLabel{Lem.~\ref{lem:extended-lint-properties}-(ii)}
\dashedLine
\UnaryInfC{$\R, w \leq u, u \leq v, a \in D_{v}, v : A[a/x], w : \forall x A, u : \forall x A, \Gamma \Rightarrow \Delta$}
\RightLabel{$(\forall_{l})$}
\UnaryInfC{$\R, w \leq u, u \leq v, a \in D_{v}, w : \forall x A, u : \forall x A, \Gamma \Rightarrow \Delta$}
\RightLabel{$\lift$}
\UnaryInfC{$\R, w \leq u, u \leq v, a \in D_{v}, w : \forall x A,  \Gamma \Rightarrow \Delta$}
\DisplayProof
\end{center}

\begin{center}
\AxiomC{$\R, u \leq z, z \leq z', u \leq z', a \in D_{v}, w : A[a/x], w : \forall x A, \Gamma \Rightarrow \Delta$}
\RightLabel{$(\forall^{*}_{l})$}
\UnaryInfC{$\R, u \leq z, z \leq z', u \leq z', a \in D_{v}, w : \forall x A, \Gamma \Rightarrow \Delta$}
\RightLabel{$\trans$}
\UnaryInfC{$\R, u \leq z, z \leq z', a \in D_{v}, w : \forall x A, \Gamma \Rightarrow \Delta$}
\DisplayProof
\end{center}
\begin{center}
\AxiomC{}
\RightLabel{IH}
\dashedLine
\UnaryInfC{$\R, u \leq z, z \leq z', a \in D_{v}, w : A[a/x], w : \forall x A, \Gamma \Rightarrow \Delta$}
\RightLabel{$(\forall^{*}_{l})$}
\UnaryInfC{$\R, u \leq z, z \leq z', u \leq z', a \in D_{v}, w : \forall x A, \Gamma \Rightarrow \Delta$}
\DisplayProof
\end{center}

\begin{center}
\AxiomC{$\R, u \leq z, z \leq z', u \leq z', a \in D_{v}, \Gamma \Rightarrow w : A[a/x], w : \exists x A, \Delta$}
\RightLabel{$(\exists^{*}_{r})$}
\UnaryInfC{$\R, u \leq z, z \leq z', u \leq z', a \in D_{v}, \Gamma \Rightarrow w : \exists x A, \Delta$}
\RightLabel{$\trans$}
\UnaryInfC{$\R, u \leq z, z \leq z', a \in D_{v}, \Gamma \Rightarrow w : \exists x A, \Delta$}
\DisplayProof
\end{center}
\begin{center}
\AxiomC{}
\RightLabel{IH}
\dashedLine
\UnaryInfC{$\R, u \leq z, z \leq z', a \in D_{v}, \Gamma \Rightarrow w : A[a/x], w : \exists x A, \Delta$}
\RightLabel{$(\exists^{*}_{l})$}
\UnaryInfC{$\R, u \leq z, z \leq z', a \in D_{v}, \Gamma \Rightarrow  w : \exists x A, \Delta$}
\DisplayProof
\end{center}

\begin{center}
\begin{tabular}{c}
\AxiomC{$\R, w \leq u, u \leq v, w \leq v, a \in D_{w}, a \in D_{v}, \Gamma \Rightarrow \Delta$}
\RightLabel{$\nd$}
\UnaryInfC{$\R, w \leq u, u \leq v, w \leq v, a \in D_{w}, \Gamma \Rightarrow \Delta$}
\RightLabel{$\trans$}
\UnaryInfC{$\R, w \leq u, u \leq v, a \in D_{w}, \Gamma \Rightarrow \Delta$}
\DisplayProof
\end{tabular}
\end{center}

\begin{center}
\begin{tabular}{c}
\AxiomC{$\R, w \leq u, u \leq v, w \leq v, a \in D_{w}, a \in D_{v}, \Gamma \Rightarrow \Delta$}
\RightLabel{Lem.~\ref{lem:extended-lint-properties}-(ii)}
\dashedLine
\UnaryInfC{$\R, w \leq u, u \leq v, w \leq v, a \in D_{w}, a \in D_{u}, a \in D_{v}, \Gamma \Rightarrow \Delta$}
\RightLabel{IH}
\dashedLine
\UnaryInfC{$\R, w \leq u, u \leq v, a \in D_{w}, a \in D_{u}, a \in D_{v}, \Gamma \Rightarrow \Delta$}
\RightLabel{$\nd$}
\UnaryInfC{$\R, w \leq u, u \leq v, a \in D_{w}, a \in D_{u}, \Gamma \Rightarrow \Delta$}
\RightLabel{$\nd$}
\UnaryInfC{$\R, w \leq u, u \leq v, a \in D_{w}, \Gamma \Rightarrow \Delta$}
\DisplayProof
\end{tabular}
\end{center}

\begin{center}
\begin{tabular}{c}
\AxiomC{$\R, w \leq u, u \leq v, w \leq v, a \in D_{w}, a \in D_{v}, \Gamma \Rightarrow \Delta$}
\RightLabel{$\cd$}
\UnaryInfC{$\R, w \leq u, u \leq v, w \leq v, a \in D_{v}, \Gamma \Rightarrow \Delta$}
\RightLabel{$\trans$}
\UnaryInfC{$\R, w \leq u, u \leq v, a \in D_{v}, \Gamma \Rightarrow \Delta$}
\DisplayProof
\end{tabular}
\end{center}

\begin{center}
\begin{tabular}{c}
\AxiomC{$\R, w \leq u, u \leq v, w \leq v, a \in D_{w}, a \in D_{v}, \Gamma \Rightarrow \Delta$}
\RightLabel{Lem.~\ref{lem:extended-lint-properties}-(ii)}
\dashedLine
\UnaryInfC{$\R, w \leq u, u \leq v, w \leq v, a \in D_{w}, a \in D_{u}, a \in D_{v}, \Gamma \Rightarrow \Delta$}
\RightLabel{IH}
\dashedLine
\UnaryInfC{$\R, w \leq u, u \leq v, a \in D_{w}, a \in D_{u}, a \in D_{v}, \Gamma \Rightarrow \Delta$}
\RightLabel{$\cd$}
\UnaryInfC{$\R, w \leq u, u \leq v, a \in D_{u}, a \in D_{v}, \Gamma \Rightarrow \Delta$}
\RightLabel{$\cd$}
\UnaryInfC{$\R, w \leq u, u \leq v, a \in D_{v}, \Gamma \Rightarrow \Delta$}
\DisplayProof
\end{tabular}
\end{center}

For the $(\forall_{l}^{*})$ and $(\exists_{r}^{*})$ rules, the side condition still holds after the application of IH. If the path from $w$ to $v$ traverses the relational atom $u \leq z'$, then the relational atoms $u \leq z, z \leq z'$ are still present in the sequent after evoking IH, ensuring that a path between $w$ and $v$ continues to exist.

\qed
\end{proof}

\begin{customlem}{\ref{lm:nd-cd-admiss}}
The rules $\nd$ and $\cd$ are admissible in the calculus $\intfocdl - \{\refl,\trans\}$.
\end{customlem}

\begin{proof} We prove the result by induction on the height of the given derivation, and argue the claim for $\nd$, since $\cd$ is similar. By Lem.~\ref{lm:structural-rules-permutation} and Lem.~\ref{lem:trans-refl-fo-elim}, we need only consider the $\idfo$, $(id_{q}^{*})$, $(\imp^{*}_{l})$, $\alll$, $(\forall_{l}^{*})$, $(\forall_{r}^{*})$, $\existsr$, $(\exists_{r}^{*})$, $\lift$ cases.\\

\textit{Base case.} We demonstrate the $\idfo$ case; the $(id_{q}^{*})$ case is similar.
\begin{center}
\begin{tabular}{c}

\AxiomC{}
\RightLabel{$\idfo$}
\UnaryInfC{$\R, u \leq w, w \leq v, \vv{a} \in D_{w}, b \in D_{u}, b \in D_{w}, w : p(\vv{a},b), \Gamma \Rightarrow \Delta, v : p(\vv{a},b)$}
\RightLabel{$\nd$}
\UnaryInfC{$\R, u \leq w, w \leq v, \vv{a} \in D_{w}, b \in D_{u}, w : p(\vv{a},b), \Gamma \Rightarrow \Delta, v : p(\vv{a},b)$}
\DisplayProof

\end{tabular}
\end{center}
The end sequent is derivable with the $(id_{q}^{*})$ and $\lift$ rules.

\textit{Inductive step.} We provide each proof below showing that $\nd$ can be permuted above $(\imp^{*}_{l})$, $\alll$, $(\forall_{l}^{*})$, $(\forall_{r}^{*})$, $\existsr$, $(\exists_{r}^{*})$, $\lift$ (we consider only the non-trivial cases and exclude the $\lift$ case since it is easily resolved). In the $(\forall_{l}^{*})$ and $(\exists_{r}^{*})$ cases, we assume that the side condition holds, i.e. there exists a path from $v$ to $w$. This side condition is still satisfied after the upward permutation of $\nd$ due to the existence of $u \leq v,a \in D_{u}$ in our sequent and our assumption that there is a path from $v$ to $w$. Furthermore, it should be pointed out that $(\forall_{l}^{*})$ and $(\exists_{r}^{*})$ are used to resolve the non-trivial $(\forall_{l})$ and $(\exists_{r})$ cases; observe that the side condition for $(\forall_{l}^{*})$ and $(\exists_{r}^{*})$ holds after $\nd$ is permuted upwards.

\begin{center}
\resizebox{\columnwidth}{!}{
\begin{tabular}{c}
\AxiomC{$\R, u \leq v, a \in D_{u}, a \in D_{v}, w : A \imp B, \Gamma \Rightarrow \Delta, w : A$}
\AxiomC{$\R, u \leq v, a \in D_{}, a \in D_{v}, w : A \imp B, w : B, \Gamma \Rightarrow \Delta$}
\RightLabel{$(\imp_{l}^{*})$}
\BinaryInfC{$\R, u \leq v, a \in D_{u}, a \in D_{v}, w : A \imp B, \Gamma \Rightarrow \Delta$}
\RightLabel{$\nd$}
\UnaryInfC{$\R, u \leq v, a \in D_{w}, w : A \imp B, \Gamma \Rightarrow \Delta$}
\DisplayProof
\end{tabular}
}
\end{center}
\begin{center}
\resizebox{\columnwidth}{!}{
\begin{tabular}{c}
\AxiomC{$\R, u \leq v, a \in D_{u}, a \in D_{v}, w : A \imp B, \Gamma \Rightarrow \Delta, w : A$}
\RightLabel{$\nd$}
\UnaryInfC{$\R, u \leq v, a \in D_{u}, w : A \imp B, \Gamma \Rightarrow \Delta, w : A$}
\AxiomC{$\R, u \leq v, a \in D_{u}, a \in D_{v}, w : A \imp B, w : B, \Gamma \Rightarrow \Delta$}
\RightLabel{$\nd$}
\UnaryInfC{$\R, u \leq v, a \in D_{u}, w : A \imp B, \Gamma \Rightarrow \Delta, w : A$}
\RightLabel{$(\imp_{l}^{*})$}
\BinaryInfC{$\R, u \leq v, a \in D_{u}, w : A \imp B, \Gamma \Rightarrow \Delta$}
\DisplayProof
\end{tabular}
}
\end{center}
\begin{center}
\scalebox{1}{
\begin{tabular}{c}
\AxiomC{$\R, u \leq w, w \leq v, a \in D_{u}, a \in D_{w}, w : \forall x A, v : A[a/x], \Gamma \Rightarrow \Delta$}
\RightLabel{$(\forall_{l})$}
\UnaryInfC{$\R, u \leq w, w \leq v, a \in D_{u}, a \in D_{w}, w : \forall x A, \Gamma \Rightarrow \Delta$}
\RightLabel{$\nd$}
\UnaryInfC{$\R, u \leq w, w \leq v, a \in D_{u}, w : \forall x A, \Gamma \Rightarrow \Delta$}
\DisplayProof
\end{tabular}
}
\end{center}
\begin{center}
\scalebox{1}{
\begin{tabular}{c}
\AxiomC{$\R, u \leq w, w \leq v, a \in D_{u}, a \in D_{w}, w : \forall x A, v : A[a/x], \Gamma \Rightarrow \Delta$}
\RightLabel{$\nd$}
\UnaryInfC{$\R, u \leq w, w \leq v, a \in D_{u}, w : \forall x A, v : A[a/x], \Gamma \Rightarrow \Delta$}
\RightLabel{$(\forall_{l}^{*})$}
\UnaryInfC{$\R, u \leq w, w \leq v, a \in D_{u}, w : \forall x A, \Gamma \Rightarrow \Delta$}
\DisplayProof
\end{tabular}
}
\end{center}
\begin{center}
\scalebox{1}{
\begin{tabular}{c}
\AxiomC{$\R, u \leq v, a \in D_{u}, a \in D_{v}, w : \forall x A, w : A[a/x], \Gamma \Rightarrow \Delta$}
\RightLabel{$(\forall_{l}^{*})$}
\UnaryInfC{$\R, u \leq v, a \in D_{u}, a \in D_{v}, w : \forall x A, \Gamma \Rightarrow \Delta$}
\RightLabel{$\nd$}
\UnaryInfC{$\R, u \leq v, a \in D_{u}, w : \forall x A, \Gamma \Rightarrow \Delta$}
\DisplayProof
\end{tabular}
}
\end{center}
\begin{center}
\scalebox{1}{
\begin{tabular}{c}
\AxiomC{$\R, u \leq v, a \in D_{u}, a \in D_{v}, w : \forall x A, w : A[a/x], \Gamma \Rightarrow \Delta$}
\RightLabel{$\nd$}
\UnaryInfC{$\R, u \leq v, a \in D_{u}, w : \forall x A, w : A[a/x], \Gamma \Rightarrow \Delta$}
\RightLabel{$(\forall_{l}^{*})$}
\UnaryInfC{$\R, u \leq v, a \in D_{u}, w : \forall x A, \Gamma \Rightarrow \Delta$}
\DisplayProof
\end{tabular}
}
\end{center}
\begin{center}
\AxiomC{$\R, u \leq v, b \in D_{u}, b \in D_{v}, a \in D_{w}, \Gamma \Rightarrow w : A[a/x], \Delta$}
\RightLabel{$(\forall^{*}_{r})$}
\UnaryInfC{$\R, u \leq v, b \in D_{u}, b \in D_{v}, \Gamma \Rightarrow w : \forall x A, \Delta$}
\RightLabel{$\nd$}
\UnaryInfC{$\R, u \leq v, b \in D_{u}, \Gamma \Rightarrow w : \forall x A, \Delta$}
\DisplayProof
\end{center}
\begin{center}
\AxiomC{$\R, u \leq v, b \in D_{u}, b \in D_{v}, a \in D_{w}, \Gamma \Rightarrow w : A[a/x], \Delta$}
\RightLabel{$\nd$}
\UnaryInfC{$\R, u \leq v, b \in D_{u}, a \in D_{w}, \Gamma \Rightarrow w : A[a/x], \Delta$}
\RightLabel{$(\forall_{r}^{*})$}
\UnaryInfC{$\R, u \leq v, b \in D_{u}, \Gamma \Rightarrow w : \forall x A, \Delta$}
\DisplayProof
\end{center}
\begin{center}
\scalebox{1}{
\begin{tabular}{c}
\AxiomC{$\R, u \leq w, a \in D_{u}, a \in D_{w}, \Gamma \Rightarrow \Delta, w : A[a/x], w : \exists x A$}
\RightLabel{$\existsr$}
\UnaryInfC{$\R, u \leq w, a \in D_{u}, a \in D_{w}, \Gamma \Rightarrow \Delta, w : \exists x A$}
\RightLabel{$\nd$}
\UnaryInfC{$\R, u \leq w, a \in D_{u}, \Gamma \Rightarrow \Delta, w : \exists x A$}
\DisplayProof
\end{tabular}
}
\end{center}
\begin{center}
\scalebox{1}{
\begin{tabular}{c}
\AxiomC{$\R, u \leq w, a \in D_{u}, a \in D_{w}, \Gamma \Rightarrow \Delta, w : A[a/x], w : \exists x A$}
\RightLabel{$\nd$}
\UnaryInfC{$\R, u \leq w, a \in D_{u}, \Gamma \Rightarrow \Delta, w : A[a/x], w : \exists x A$}
\RightLabel{$(\exists_{r}^{*})$}
\UnaryInfC{$\R, u \leq w, a \in D_{u}, \Gamma \Rightarrow \Delta, w : \exists x A$}
\DisplayProof
\end{tabular}
}
\end{center}
\begin{center}
\scalebox{1}{
\begin{tabular}{c}
\AxiomC{$\R, u \leq v, a \in D_{u}, a \in D_{v}, \Gamma \Rightarrow \Delta, w : A[a/x], w : \exists x A$}
\RightLabel{$(\exists_{r}^{*})$}
\UnaryInfC{$\R, u \leq v, a \in D_{u}, a \in D_{v}, \Gamma \Rightarrow \Delta, w : \exists x A$}
\RightLabel{$\nd$}
\UnaryInfC{$\R, u \leq v, a \in D_{u}, \Gamma \Rightarrow \Delta, w : \exists x A$}
\DisplayProof
\end{tabular}
}
\end{center}
\begin{center}
\scalebox{1}{
\begin{tabular}{c}
\AxiomC{$\R, u \leq v, a \in D_{u}, a \in D_{v}, \Gamma \Rightarrow \Delta, w : A[a/x], w : \exists x A$}
\RightLabel{$\nd$}
\UnaryInfC{$\R, u \leq v, a \in D_{u}, \Gamma \Rightarrow \Delta, w : A[a/x], w : \exists x A$}
\RightLabel{$(\exists_{r}^{*})$}
\UnaryInfC{$\R, u \leq v, a \in D_{u}, \Gamma \Rightarrow \Delta, w : \exists x A$}
\DisplayProof
\end{tabular}
}
\end{center}
\qed
\end{proof}

\newpage

\section{The Nested Calculi $\nint^{*}$ and $\nintfocd^{*}$}
\label{app:new-nested-calculi}

\begin{figure}
\noindent\hrule

\begin{center}
\begin{tabular}{c c c}
\AxiomC{} \RightLabel{$\id$}
\UnaryInfC{$\Sigma \lcut X, p \far p, Y \rcut$}
\DisplayProof

&

\AxiomC{$\Sigma \lcut X, A,B \far Y \rcut $}
\RightLabel{$\conrl$}
\UnaryInfC{$\Sigma \lcut X, A \land B \far Y \rcut$}
\DisplayProof

&

\AxiomC{$\Sigma \lcut X \far A,B, Y \rcut $}
\RightLabel{$\disrr$}
\UnaryInfC{$\Sigma \lcut X \far A\lor B, Y \rcut$}
\DisplayProof
\end{tabular}
\end{center}

\begin{center}
\begin{tabular}{c c}
\AxiomC{$\Sigma \lcut X, A \far Y \rcut$}
\AxiomC{$\Sigma \lcut X, B \far Y \rcut$}
\RightLabel{$\disrl$}
\BinaryInfC{$\Sigma \lcut X, A \lor B \far Y \rcut$}
\DisplayProof

&

\AxiomC{$\Sigma \lcut X \far A, Y \rcut$}
\AxiomC{$\Sigma \lcut X \far B, Y \rcut$}
\RightLabel{$\conrr$}
\BinaryInfC{$\Sigma \lcut X \far A\land B, Y \rcut$}
\DisplayProof

\end{tabular}
\end{center}

\begin{center}
\begin{tabular}{c c c}
\AxiomC{$\Sigma \lcut X \far Y, [A \far ] \rcut$}
\RightLabel{$(\neg_{r})$}
\UnaryInfC{$\Sigma \lcut X \far Y, \neg A \rcut$}
\DisplayProof

&

\AxiomC{$\Sigma \lcut X, \neg A \far A, Y \rcut$}
\RightLabel{$(\neg_{l})$}
\UnaryInfC{$\Sigma \lcut X, \neg A \far Y \rcut$}
\DisplayProof

&

\AxiomC{$\Sigma\{X, A \far Y, [X', A \far Y']\}$}
\RightLabel{$\lift$}
\UnaryInfC{$\Sigma\{X, A \far Y, [X' \far Y']\}$}
\DisplayProof
\end{tabular}
\end{center}

\begin{center}
\begin{tabular}{c c}
\AxiomC{$\Sigma \lcut X \far Y, [A \far B] \rcut$}
\RightLabel{$\imprr$}
\UnaryInfC{$\Sigma \lcut X \far A \imp B, Y \rcut$}
\DisplayProof

&

\AxiomC{$\Sigma \lcut X, A \imp B \far A, Y \rcut$}
\AxiomC{$\Sigma \lcut X, A \imp B, B \far Y \rcut$}
\RightLabel{$\imprl$}
\BinaryInfC{$\Sigma \lcut X, A \imp B \far Y \rcut$}
\DisplayProof
\end{tabular}
\end{center}

\begin{center}
\begin{tabular}{c c}
\AxiomC{}
\RightLabel{$(id_{q})$}
\UnaryInfC{$\Sigma \lcut X, p(\vv{a}) \far p(\vv{a}), Y \rcut$}
\DisplayProof

&

\AxiomC{$\Sigma \lcut X \far Y, A[a/x], \exists x A  \rcut$}
\RightLabel{$(\exists_{r})$}
\UnaryInfC{$\Sigma \lcut X \far Y, \exists x A \rcut$}
\DisplayProof
\end{tabular}
\end{center}

\begin{center}
\begin{tabular}{c c c} 
 
\AxiomC{$\Sigma \lcut X \far Y, A[a/x]  \rcut$}
\RightLabel{$(\forall_{r})^{\dag}$}
\UnaryInfC{$\Sigma \lcut X \far Y, \forall x A \rcut$}
\DisplayProof

&

\AxiomC{$\Sigma \lcut X, A[a/x], \forall x A \far Y  \rcut$}
\RightLabel{$(\forall_{l})$}
\UnaryInfC{$\Sigma \lcut X, \forall x A \far Y \rcut$}
\DisplayProof

&

\AxiomC{$\Sigma \lcut X, A[a/x] \far Y  \rcut$}
\RightLabel{$(\exists_{l})^{\dag}$}
\UnaryInfC{$\Sigma \lcut X, \exists x A \far Y \rcut$}
\DisplayProof
\end{tabular}
\end{center}

\hrule
\caption{The nested calculus $\nint^{*}$ for propositional intuitionistic logic consists of the first four lines, and all rules taken together give the nested calculus $\nintfocd^{*}$. The side condition $\dag$ states that $a$ does not occur in the conclusion.}
\label{fig:nested-calculus-propositional-APP}
\end{figure}

\end{document}